\documentclass[a4paper,UKenglish,cleveref, autoref,thm-restate]{lipics-v2021}
\pdfoutput=1 %uncomment to ensure pdflatex processing (mandatatory e.g. to submit to arXiv)
\hideLIPIcs  %uncomment to remove references to LIPIcs series (logo, DOI, ...), e.g. when preparing a pre-final version to be uploaded to arXiv or another public repository

\title{Parallel Query Processing with Heterogeneous Machines} 

\author{Simon Frisk}{University of Wisconsin-Madison}{simon.frisk@wisc.edu}{}{}%TODO mandatory, please use full name; only 1 author per \author macro; first two parameters are mandatory, other parameters can be empty. Please provide at least the name of the affiliation and the country. The full address is optional. Use additional curly braces to indicate the correct name splitting when the last name consists of multiple name parts.

\author{Paraschos Koutris}{University of Wisconsin-Madison}{paris@cs.wisc.edu}{}{}

\authorrunning{S. Frisk and P. Koutris} %TODO mandatory. First: Use abbreviated first/middle names. Second (only in severe cases): Use first author plus 'et al.'

\Copyright{Simon Frisk and Paraschos Koutris} %TODO mandatory, please use full first names. LIPIcs license is "CC-BY";  http://creativecommons.org/licenses/by/3.0/



\nolinenumbers

\usepackage{graphicx}
\usepackage{amsthm}
\usepackage{amsmath}
\usepackage{amssymb}
\usepackage{float}
\usepackage{xcolor}
\usepackage{pgfplots}
\usepackage{algpseudocode}
\usepackage{algorithm}

\newcommand{\introparagraph}[1]{\vspace{0.7mm} \noindent \textbf{#1.}}

\newcommand{\dl}{\textit{ :- }}
\newcommand{\norm}[2]{\left\lVert #1 \right\rVert_{#2}}

\newcommand{\vol}[1]{|{#1}|}
\newcommand{\bu}{\mathbf{u}}
\newcommand{\bv}{\mathbf{v}}
\newcommand{\bw}{\mathbf{w}}

\begin{document}

\maketitle

\begin{abstract}
    We study the problem of computing a full Conjunctive Query in parallel using $p$ heterogeneous machines. Our computational model is similar to the MPC model, but each machine has its own cost function mapping from the number of bits it receives to a cost. An optimal algorithm should minimize the maximum cost across all machines. We consider algorithms over a single communication round and give a lower bound and matching upper bound for databases where each relation has the same cardinality. We do this for both linear cost functions like in previous work, but also for more general cost functions. For databases with relations of different cardinalities, we also find a lower bound, and give matching upper bounds for specific queries like the cartesian product, the join, the star query, and the triangle query. Our approach is inspired by the HyperCube algorithm, but there are additional challenges involved when machines have heterogeneous cost functions.
\end{abstract}

\section{Introduction}

Large datasets are commonly processed using massively parallel systems. To analyze query processing in such a setting,  Beame et al.~\cite{communication_steps_journal} introduced the \textit{massively parallel computation} (MPC) model. The MPC model considers a cluster with a shared-nothing architecture, where computation proceeds in rounds: each round consists of communication between machines, followed by computation on the locally stored data. The main measure of complexity in the MPC model is the {\em load}, which captures the maximum number of bits received by a machine. An efficient MPC algorithm is designed to make the load as small as possible.
%Each round is followed by a global synchronization step. The MPC model measures complexity using two measures - the number of rounds $r$ and the maximum number of bits received by a machine, $L$. This reflects that in massively distributed systems, the performance bottleneck lies in communication between machines and synchronization.

However, the MPC model operates on an assumption of homogeneity, meaning the cost of a machine is indifferent to where the received data was sent from, and how powerful the machine is. This is an unrealistic assumption, as the large-scale clusters that massively parallel computation is performed on are {\em heterogeneous}. Heterogeneity can occur in both compute resources (processing speed, memory) and the network that connects the machines. %for example, if machines are on different racks in a data center, where inter-rack communication is cheaper than communication across racks.
%Heterogeneity can also change over time, for example, when network congestion occurs.

In this work, we consider massively parallel data processing in clusters with heterogeneity in compute resources. We use a computational model that, similar to the MPC model, has a homogeneous network topology (every machine is connected directly to any other machine). However, each machine $c$ is equipped with its own cost function $g_c$: this function maps the number of bits the machine receives to the cost. The load $L$ of a round is then defined as the maximum cost across all machines, i.e., $L = \max_{c\in[p]} g_c$. The computational model in this paper captures the MPC model as a special case, when for each machine the cost function is the identity function $g_c(N)=N$. Our model is also a special instance of the topology-aware model in~\cite{topologyawaredataprocessing}, however, one that has not been studied in prior work.

Based on the above heterogeneous model, we study the problem of computing join queries with a minimum load. We will focus on one-round algorithms, i.e., we want to have only local computation after one round of communication. One-round algorithms are particularly relevant to data processing systems with a disaggregated storage architecture (e.g., Amazon Aurora~\cite{amazonaurora}, Snowflake~\cite{snowflake}).
%In Snowflake, data is stored in Amazon S3, and query processing is done in virtual warehouses, which are clusters that scale to hundreds of nodes and thousands of cores. 
These algorithms can be viewed as algorithms that send the data from the storage layer to the compute layer in such a way that no further communication has to be done in the compute layer. This paper therefore addresses the problem of optimally sending data from the data layer to the compute layer when there is compute heterogeneity.

%Previous work has focused mainly on analyzing data processing in homogenous clusters\cite{communication_steps_journal,beame2014skewparallelqueryprocessing,beame2016worstcaseoptimalalgorithmsparallel,ketzman,coverorpack}.

\introparagraph{Our Contributions}
The main contribution of this work is upper and lower bounds for the load $L$ of computing a join query (corresponding to a full Conjunctive Query) in one round with heterogeneous machines. In particular:
\begin{itemize}
\item We present an algorithm (Section~\ref{sec:upper}) that evaluates a join query in one round when the cost function is linear with different weights, i.e., $g_c(N) = N/w_c$ for machine $c$. Our algorithm works for two different types of inputs where all relations have the same size: matching databases that are sparse, and dense databases that contain a constant fraction of all possible input tuples.

\item We give (Section~\ref{sec:lower}) lower bounds that (almost) match the upper bounds for both the sparse and dense cases. Our lower bounds are unconditional, that is, they make no assumptions on how the algorithm behaves and how it encodes the input tuples. 

\item We next consider the case with non-linear cost functions (Section~\ref{sec:general}). Previous work, even in the topology-aware MPC model, assumes linear cost functions. We generalize this to a wider class of cost functions.

\item Finally, we consider queries where the cardinalities of input relations are different (Section~\ref{sec:unequal}). We give a lower bound on the load to compute such queries in a single round, for the same two data distributions as in the equal cardinality case. We also give an algorithm that matches the upper bound for Conjunctive Queries for the cartesian product, binary join, star query, and triangle query. 
\end{itemize}

\introparagraph{Technical Ideas}
In the MPC model, the HyperCube algorithm has proved to be the key technique that gives optimal join algorithms. The HyperCube algorithm maps tuples to machines via a hash function that hashes each tuple to a vector. Tuples are sent to machines where the projection of the coordinates of the machine equals the hash vector of the tuple. Each machine obtains the same number of tuples (with high probability) and has the same load. However, in the heterogenous setting, each machine may be allocated a different number of tuples, since slower machines can handle less data than faster machines. Thus, instead of considering how to organize the machines in a hypercube, we consider how to partition the space of all possible tuples $\Lambda=[n]^k$ into subspaces (which are hyperrectangles) $\Lambda_c\subseteq\Lambda$, one for each machine $c$. Each machine is then responsible for computing all the output tuples in this subspace, and to do this correctly it needs to receive all input tuples that may contribute to these. The technical challenge is twofold: $(i)$ how to optimally set the dimensions of each $\Lambda_c$ to minimize the load across all machines, and $(ii)$ how to geometrically position the subspaces such that the space $\Lambda$ is fully covered. We will show that query parameters such as fractional edge packings and vertex covers are still critical in characterizing the optimal load, but the algorithmic techniques we use are different from the HyperCube algorithm.

\section{Related Work}

\introparagraph{MPC Algorithms}
The MPC model is a computational model introduced by Beame et al.~\cite{communication_steps_journal}. It has been used to analyze parallel algorithms for joins and other fundamental data processing tasks. The seminal paper~\cite{communication_steps_journal} shows matching upper and lower bounds on the load for Conjunctive Queries in one round for matching databases. A lower bound for queries with skew was also given, which was matched by an upper bound for some classes of queries. Later work~\cite{beame2016worstcaseoptimalalgorithmsparallel} studied the worst-case optimal load for any input in one round algorithms and proposed an algorithm matching the lower bound. Further research explored the computation of join queries using multiple rounds~\cite{ beame2016worstcaseoptimalalgorithmsparallel, ketzman, coverorpack, parallelAcyclicJoins, Tao2020}, or the design of parallel output-sensitive algorithms in the MPC model~\cite{Hu_2019_OutputSensitive}. 

\introparagraph{Topology-aware Algorithms}
A recent line of work aims to consider a topology-aware parallel model that is aware of the heterogeneity in the cluster topology and compute resources~\cite{topologyawaredataprocessing, hu2020algorithms,topologyawarejoins}.
In this model, the topology is modeled as a graph $G=(V,E)$, where a subset $V_C\subseteq V$ of nodes are compute nodes. Computation proceeds in rounds similar to the MPC model, but the cost model is different. Instead of modeling the cost as the maximum number of bits sent to a processor, each edge in the network has a cost which is a function of the number of bits it transmits. The cost of a round is then the maximum cost across all edges. A common cost function is that the cost of edge $e$ is $f_e(N)=N/w_e$, which is similar to the cost function used in this paper. 
Under this topology-aware model, recent work has studied lower and upper bounds for set intersections, cartesian product, and sorting~\cite{hu2020algorithms}, as well as binary joins~\cite{topologyawarejoins}. Both of these papers assume that the underlying network has a symmetric tree topology.

%The new model models the cluster as a directed graph, where a subset of nodes can perform computation, and the cost of transmitting one bit depends on how the bit traverses the cluster graph. Like in the MPC model, the complexity measures are the number of rounds $r$ and the load $L$. $L$ is here the cost of the maximum link in the graph, which is a function of the number of bits the link transmits.

The computational model in this paper is a special case of the topology-aware MPC model, where the network topology is a star. This is a tree with depth 1, where all leaves are compute nodes, and the root node is a router. The cost function from a compute node to the router is $0$, and the cost function from the router to machine $c$ is precisely the cost function of the machine, $g_c(N)$. Prior work in the topology-aware MPC model does not capture the work in this paper, for two reasons. First, it considers symmetric trees, meaning the cost function across a link is the same in each direction, which is not true in this paper. Second, we consider arbitrary full conjunctive queries, which have not been studied previously.

\section{Background}

\introparagraph{Computation Model}
Initially, the $p$ machines in the cluster hold an arbitrary piece of the input data. The computation then proceeds in $r$ rounds. A round consists of the {\em communication phase}, where machines can exchange data, followed by the {\em computation phase}, where computation is performed on locally stored data. In this paper, we focus on algorithms where $r=1$, meaning there is a single round of communication followed by computation on local data. The output of a computation is the (set) union of the output across all machines.

In the standard MPC model, the cost of a round is modeled as the maximum amount of data (in bits) received by any machine. That is, if $N_c$ is the number of bits received by machine $c$, the cost of computation is $L=\max_{c\in[p]}N_c$. 

In this paper, we will extend this model to heterogeneous machines. This means that each machine $c \in [p]$ has a cost function $g_c:\mathbb{Z}^+\rightarrow \mathbb{R}^+$ that maps from the number of bits received ($N_c$) to a positive real number denoting cost. The cost of a round is similar to before, i.e., $\max_{c\in[p]} g_c(N_c)$. We will mostly work with linear cost functions $g_c(x)=x/w_c$ for some $w_c\in \mathbb{Z^+}$. Here, the weight constant $w_c$ for each machine captures the cost at the machine, which may include both data transmission and processing. Later in the paper, we will study more general cost functions.

\introparagraph{Conjunctive Queries}
In this paper, we work with Conjunctive Queries without projection or selection. These can be thought of as natural joins between $l$ relations:
$$ q(x_1,...,x_k)\dl S_1(\mathbf{y}_1),...,S_l(\mathbf{y}_l) $$

There are $k$ variables, denoted $x_1,...,x_k$, and $l$ atoms, denoted $S_1,...,S_l$. For each $j$, the vector $\mathbf{y}_j$ consists of variables, and $r_j$ is the arity of the atom $S_j$. 
We restrict the queries in this paper to have no self-joins, meaning no two atoms can refer to the same underlying relation. We will often use the notation $x \in S_j$ to mean that variable $x$ occurs in the atom $S_j$. We will work with relations where the values come from a domain $[n]$ = \{1,2, \dots, n\}. We denote the cardinality of atom $S_j$ as $m_j$ and the number of bits needed to encode $S_j$ as $M_j$.

A {\em fractional vertex cover} $\bv$ for $q$ assigns a weight $v_i \geq 0$ to each variable $x_i$ such that for every atom $S_j$, we have $\sum_{x_i \in S_j} v_i \geq 1$.

A {\em fractional edge packing} $\bu$ for $q$ assigns a weights $u_i \geq 0$ to each atom $S_j$ such that for every variable $x_i$, we have $\sum_{j:x_i\in S_j}u_j \leq 1$.

\introparagraph{HyperCube Algorithm}
HyperCube is an elegant algorithm for distributed multiway joins, originally introduced by Afrati and Ullman for the MapReduce model \cite{afratiullman}. It computes multiway joins in a single round of communication, as opposed to traditional methods where relations are joined pairwise. We will illustrate how HyperCube computes a full CQ $q$ with $k$ variables using $p$ machines. 

The $p$ machines are organized in a hyperrectangle with $k$ dimensions, one for each variable. The sides of the hyperrectangle have $\{p_i\}_{i\in[k]}$ machines, where $p_i\in[1,p]$ and $\prod_{i\in[k]}p_i=p$. Each machine $c$ has a coordinate $\mathbf{C}_c\in[p_1]\times...\times[p_k]$. Denote $\pi_{S_j}\mathbf{C}_c$ as the projection of $\mathbf{C}_c$ on $S_j$. We will use $k$ hash functions $\{h_i\}_{i\in[k]}$, one for each variable, where $h_i:[n]\rightarrow [p_i]$. Denote $\mathbf{h}=(h_1,...,h_k)$ as the vector of all hash functions, and $\pi_{S_j}\mathbf{h}$ as the projection of $\mathbf{h}$ on $S_j$.
A tuple $a_j\in S_j$ will be sent to all machines $c$ such that $(\pi_{S_j}\mathbf{h})(a_j)=\pi_{S_j}\mathbf{C}_c$. Then the query can be computed locally on each machine with all tuples that were sent to that machine. The correctness of the algorithm follows from that each tuple $\mathbf{a}\in[n]^k$ that should be in the output is produced by the machine $\mathbf{h}(\mathbf{a})$.

\introparagraph{Input Distributions}
In this paper, we will focus on two classes of inputs, sparse and dense. The first type of input is a {\em matching database}. The cardinality of relation $S_j$ is $m_j$. For every value in the domain $v\in[n]$, every relation $S_j$, and attribute $A$ of that relation, there exists at most one tuple $a_j\in S_j$ such that the value of $a_j$ in the attribute $A$ is $v$. If the arity of a relation $S_j$ is $1$, we require that $m_j/n\leq\theta$ for some constant $\theta\in(0,1)$. We will start by considering the case when each relation has the same cardinality. In Section~\ref{sec:unequal}, we will generalize this to the case when each relation can have a different cardinality $m_j\leq n$.

%For a relation $S_j$ where values are from the domain $[n]$ with arity $r_j$, there are $n^{r_j}$ possible tuples. 
The second class of inputs are $\theta$-dense databases, where $\theta \in (0,1)$. For this input, a relation $S_j$ or arity $r_j$ has a fraction $\theta$ of all $n^{r_j}$ possible tuples. We consider $\theta$ to be a constant in data complexity terms. We will first study instances where the cardinality of each relation is the same (which means that the arity $r_j$ is the same for each relation) and generalize in Section \ref{sec:unequal} to unequal cardinalities.

\section{The Upper Bound}
\label{sec:upper}

In this section, we give algorithms for computing a full Conjunctive Query $q$ with $k$ variables. 
%Consider an instance $I$ with values from domain $[n]$. Let $q$ be a full CQ with $k$ variables $x_1, \dots, x_k$. 
We will consider the linear cost model, where we have  $p$ machines, and machine $c \in [p]$ has a linear cost function $g_c(N) = N/w_c$ for some weight $w_c \geq 0$. We will denote $\bw := (w_1, \dots, w_p)$. 

Let $I$ be an instance with uniform cardinalities $m$ over a domain $[n]$. Let $\bv$ be a fractional vertex cover of $q$ and $v = \sum_{i\in[k]} v_i$. Then, define:
$$L^{\textsf{upper}}_\bv := \frac{m \log n}{\norm{\bw}{v}} = \frac{m \log n}{\left( \sum_{c \in [p]} w_c^v \right)^{1/v}}$$

\begin{theorem}[Dense Inputs]
\label{theorem:UpperBoundDense}
Let $q$ be a full CQ with uniform arity $r$ and a $\theta$-dense input $I$ with domain $[n]$ (every relation has size $m=\theta n^r$). Then, for every fractional vertex cover $\bv$, we can evaluate $q$ in one round in the linear cost model with load $O(L^{\textsf{upper}}_\bv)$.
\end{theorem}

\begin{theorem}[Sparse Inputs]
\label{theorem:UpperBoundSparse}
Let $q$ be a full CQ and $I$ be a matching database with domain $[n]$ and uniform relation sizes $m$. Then, for every fractional vertex cover $\bv$ we can evaluate $q$ in one round in the linear cost model with load (with high probability) $O(L^{\textsf{upper}}_\bv)$.
\end{theorem}

In the rest of the section, we will prove the above two theorems. We start with an overview of our approach, which is similar to the HyperCube algorithm albeit with some important modifications. We do not consider how to pick share exponents to decide the number of machines to put in each dimension of the hypercube. This concept is now not meaningful, since the machines are different. 

Instead, we consider the hyperrectangle $\Lambda=[n]^k$, which can be thought of as the space containing all possible output tuples.
Our algorithm partitions $\Lambda$ into hyperrectangles $\{\Lambda_c\}_{c \in [p]}$. We will use this partitioning to guide how machines will compute the output. To do this, we need a vector of $k$ functions $\mathbf{h}=(h_1,\dots,h_k)$, where $h_i:[n]\rightarrow[n]$. For the sparse data distribution, $\mathbf{h}$ will be a random hash function (essentially perturbing the input tuples). For the dense data distribution, $\mathbf{h}$ will be the identity function $\mathbf{h}(\mathbf{a})=\mathbf{a}$. 

Then, machine $c$ will be responsible for computing every tuple $\mathbf{a}\in[n]^k$ such that $\mathbf{h}(\mathbf{a})\in\Lambda_c$. To achieve this, our algorithm sends information about a tuple $a_j\in S_j$ to all machines $c$ where $(\pi_{S_j}\mathbf{h})(a_j)\in \pi_{S_j}\Lambda_c$, where $\pi_{S_j} \Lambda_c$ is the projection of the subspace to the attributes of $S_j$. Similar to the HyperCube algorithm, this guarantees that every potential output tuple $\mathbf{a}$, if it exists in the output, is produced at one machine, namely the machine $c$ with $\mathbf{h}(\mathbf{a})\in \Lambda_c$.

We will denote by $\lambda_{c,i}$ the side length of $\Lambda_c$ on variable $x_i$ for machine $c$. Moreover, we will use $\vol{\Lambda}$ to denote the volume of $\Lambda$, i.e., the number of points in the space. Note that $\vol{\pi_S\Lambda_c}=\prod_{x\in S}\lambda_{c,i}$.

There are two main aspects to describe of our algorithm. The first is how to pick the side lengths $\lambda_{c,i}$ for each machine and dimension to minimize the load -- this corresponds to minimizing the projections $\pi_{S_j} \Lambda_c$ of the hyperrectangles. The second is how to geometrically position the hyperrectangles $\Lambda_c$ in $\Lambda$ to cover the whole space. We describe these two components in the next two sections.

\subsection{Partitioning the Space}

\begin{theorem}
    \label{theorem:partitionspace}
    Let $\mathbf{v}=(v_1,...,v_k)$ be any fractional vertex cover of a CQ $q$. Let $v=\sum_{j \in [k]}v_i$. For every machine $c$, let the side length of a hyperrectangle $\Lambda_c$ in $\Lambda$ along some variable $x_i$ be $$ \lambda_{c,i} :=\left(\frac{w_c}{\norm{\mathbf{w}}{v}}\right)^{v_i}n $$
    Then, the following two properties hold:
    \begin{enumerate}
        \item $\sum_{c\in[p]} |\Lambda_c| = n^k$;
        \item for every machine $c$ and every atom $S$ with arity $r$: $\vol{\pi_{S} \Lambda_c} \leq \frac{w_c}{\norm{\mathbf{w}}{v}} \cdot n^r $
    \end{enumerate}
\end{theorem}
\begin{proof}
    We start by showing that the assignment above covers all of $\Lambda$, by summing the covered volume for each machine.
    \[
    \sum_{c\in[p]} \vol{\Lambda_c} =
    \sum_{c\in[p]}\prod_{j\in[k]}\lambda_{c,i}
    =\sum_{c\in[p]}\prod_{j\in[k]}\left[\left(\frac{w_c}{\norm{\mathbf{w}}{v}}\right)^{v_i}n\right]
    =\sum_{c\in[p]}\left[\left(\frac{w_c}{\norm{\mathbf{w}}{v}}\right)^vn^k\right]
    =n^k\frac{\sum_{c\in[p]}w_c^v}{\sum_{c\in[p]}w_c^v}=n^k
    \]
    Next, we show the bound on the volume of the projected hyperrectangle on each atom. We focus on some atom $S$ with arity $r$. Then, we have:
    $$
    \vol{\pi_{S} \Lambda_c} = \prod_{x_i \in S} \lambda_{c,i} =
    \left(\frac{w_c}{\norm{\mathbf{w}}{v}}\right)^{\sum_{x_i\in S}v_i}n^r
    $$
    Note that $\frac{w_c}{\norm{\bw}{v}}\leq 1$. Furthermore, since $\mathbf{v}$ is a vertex cover, $\sum_{x_i\in S}v_i \geq 1$. Hence, we get the desired inequality.
\end{proof} 

The above lemma provides the appropriate dimensions of each hyperrectangle $\Lambda_c$, but it does not tell us how these hyperrectangles must be positioned geometrically within $\Lambda$ such that they cover the whole space.

\begin{example}
    Consider the Cartesian product $q(x,y) \dl S_1(x),S_2(y)$. We have $p=17$ machines. There are 2 machines with $w=4$, 1 machine with $w=3$, 3 machines with $w=2$, and 11 machines with $w=1$. Consider the vertex cover with $v_x =v_y=1$. Then, $\norm{\bw}{u}=8$. This gives that machines with $w=4$ should have side lengths $n/2$, machines with $w=3$ should have side lengths $3n/8$, machines with $w=2$ side lengths $n/4$ and finally $w=1$ should have side lengths $n/8$. The figure below shows one way to position the rectangles to cover $\Lambda$. Each rectangle is labeled with the weight of the machine that occupies that space.
    \begin{figure}[H]
        \centering
        \scalebox{0.55}{\begin{tikzpicture}
            \draw[very thick] (0,0) rectangle (8,8);
            
            \draw[thick] (4,0) rectangle (8,4);
            \node at (6,2) {4};
            \draw[thick] (0,4) rectangle (4,8);
            \node at (2,6) {4};
            
            \draw[thick] (0,0) rectangle (3,3);
            \node at (1.5,1.5) {3};
            
            \draw[thick] (4,6) rectangle (6,8);
            \node at (5, 7) {2};
            \draw[thick] (6,4) rectangle (8,6);
            \node at (7,5) {2};
            \draw[thick] (6,6) rectangle (8,8);
            \node at (7,7) {2};
            
            \draw[thick] (0,3) rectangle (1,4);
            \node at (0.5,3.5) {1};
            \draw[thick] (1,3) rectangle (2,4);
            \node at (1.5,3.5) {1};
            \draw[thick] (2,3) rectangle (3,4);
            \node at (2.5,3.5) {1};
            \draw[thick] (3,3) rectangle (4,4);
            \node at (3.5,3.5) {1};
            \draw[thick] (3,2) rectangle (4,3);
            \node at (3.5,2.5) {1};
            \draw[thick] (3,1) rectangle (4,2);
            \node at (3.5,1.5) {1};
            \draw[thick] (3,0) rectangle (4,1);
            \node at (3.5,0.5) {1};
            
            \draw[thick] (4,4) rectangle (5,5);
            \node at (4.5,4.5) {1};
            \draw[thick] (5,4) rectangle (6,5);
            \node at (5.5,4.5) {1};
            \draw[thick] (4,5) rectangle (5,6);
            \node at (4.5,5.5) {1};
            \draw[thick] (5,5) rectangle (6,6);
            \node at (5.5,5.5) {1};
        \end{tikzpicture}}
        \caption{One way to pack the machines in the example.}
        \label{fig:packcartesianprodexample}
    \end{figure}
\end{example}

In the example above we can perfectly fit the rectangles together to cover $\Lambda$. In the case when all hyperrectangles have the same dimensions, such as when machines have the same weight $w_c$, packing is a trivial problem. In general, there might not be a perfect way to fit the hyperrectangles together to cover the full space. This will require us to increase the size of some of the hyperrectangles $\Lambda_c$, but the volumes will be increased only by a constant factor.

\subsection{Packing Hyperrectangles}
\label{subsec:Packing}

In this subsection, we will show how to geometrically position the hyperrectangles $\{\Lambda_1,...,\Lambda_p\}$ to cover $\Lambda$. During this process, we will have to adjust the dimensions of each $\Lambda_c$ so that the hyperrectangles can fit together. This will result in adjusted hyperrectangles $\{\Bar{\Lambda}_1,...,\Bar{\Lambda}_p\}$, however, we only have to pay a constant factor increase in their dimensions. In particular:

\begin{theorem}[Packing Theorem]\label{theorem:packingtheorem}
The hyperrectangles $\{\Lambda_1,...,\Lambda_p\}$ can be packed to cover $\Lambda$ by adjusting hyperrectangles to $\{\Bar{\Lambda}_1,...,\Bar{\Lambda}_p\}$ such that for all relations $S_j$ with arity $r_j$ and machines $c$, $\vol{\pi_{S_j}\Bar{\Lambda}_c} \leq 2^{k+1+r_j} \cdot \vol{\pi_{S_j}\Lambda_c}$.
\end{theorem}

Except for in this subsection, we will always denote the hyperrectangle for machine $c$ as $\Lambda_c$, even after the packing algorithm has run.

A condensed description of the packing algorithm can be seen in~\autoref{algorithm:Packing}. The algorithm sets dimensions of $\Lambda_c$ according to~\autoref{theorem:partitionspace}. Each side of each hyperrectangle is then rounded independently to the nearest higher power of two. This gives some adjusted hyperrectangles, $\{\hat{\Lambda}_1,...,\hat{\Lambda}_p\}$. The hyperrectangles are then put into buckets, where each bucket contains all hyperrectangles of the same size. Denote the number of buckets as $b$.

Because the sides of hyperrectangles have been rounded to powers of two, we can always, if we have enough hyperrectangles in some small bucket, merge them into one hyperrectangle that fits in a larger bucket. We will order buckets in increasing order of hyperrectangle size. Starting with the first bucket, we will merge as many hyperrectangles as possible into hyperrectangles that fit in the second bucket. We do this for each consecutive pair of buckets until the last bucket is reached.

In the next step, we take the largest bucket, and pairwise merge hyperrectangles into hyperrectangles of twice the volume, by stacking them in a minimum dimension. This gives a new bucket of hyperrectangles. We repeat this procedure until there is just one hyperrectangle $R$ in the obtained bucket.

We will now take this hyperrectangle $R$ and use it to fill $\Lambda$. Some dimensions of $R$ may be smaller than $n$. In such a case, we just scale up $R$ in those dimensions to be exactly $n$.

\begin{algorithm}
\caption{Packing Algorithm}
\label{algorithm:Packing}
\begin{algorithmic}[1]
    \State $\Lambda_1,\dots,\Lambda_p \gets \text{According to  }\autoref{theorem:partitionspace}$.
    \State $\hat{\Lambda}_1,\dots,\hat{\Lambda}_p \gets \text{Round each side to higher power of two}$.
    \State $B_1,\dots,B_b \gets \text{Buckets of $\hat{\Lambda}_c$ of similar dimensions}$.
    \For{$B_t\in \{B_1,\dots,B_{b-1}\}$}
        \State Merge as many rectangles from $B_t$ into $B_{t+1}$ as possible.
    \EndFor
    \State $t \gets b$
    \While{$|B_t|>1$}
        \State $B_{t+1} \gets \text{Pairwise merge hyperrectangles in $B_t$ in the smallest dimension}$.
        \State $t \gets t+1$.
    \EndWhile
    \State $R \gets \text{The one hyperrectangle in $B_t$}$.
    \State Scale $R$ up to cover $\Lambda$.
    \State\Return $R$.
\end{algorithmic}
\end{algorithm}

We now analyze the details of the algorithm. Recall that the packing algorithm starts by rounding all sides $\lambda_{c,i}$ to the nearest higher power of two, $\hat{\lambda}_{c,i}$, obtaining rounded hyperrectangles $\hat{\Lambda}_c$. 
That is, for each machine $c$ and dimension $i$ we find $\alpha_{c,i}$ such that:
$ 2^{\alpha_{c,i}-1} < \lambda_{c,i} \leq 2^{\alpha_{c,i}}=\hat{\lambda}_{c,i} $
Each side of $\Lambda_c$ is rounded independently. This means that for $\hat{\Lambda}_c$, $\hat{\Lambda}_{c'}$ where $i\neq i'$, it is possible that $\lambda_{c,i}=\lambda_{c',i}$ for some but not all variables $x_i$.

\begin{lemma}
    \label{lemma:machineordering}
    For any two machines with weights $w_c$ and $w_{c'}$ such that $w_c\leq w_{c'}$, for any variable $x_i$, $\hat{\lambda}_{c,i}\leq \hat{\lambda}_{c',i}$.
\end{lemma}

\begin{proof}
    Since $w_c\leq w_{c'}$, we know that $
    \left(\frac{w_c}{\norm{\bw}{v}}\right)^{v_i}n
    \leq \left(\frac{w_{c'}}{\norm{\bw}{v}}\right)^{v_i}n$.
    This means that $\lambda_{c,i}\leq\lambda_{c',i}$. Then it is also true that $\hat{\lambda}_{c,i}\leq \hat{\lambda}_{c',i}$.
\end{proof}

We now create buckets $B_1,...,B_b$ of all hyperrectangles $\{\hat{\Lambda}_c\}_{c\in[p]}$, one bucket for each hyperrectangle with the same dimensions. This means that for each hyperrectangle $\hat{\Lambda}_c, \hat{\Lambda}_{c'}$ in the same bucket, for all $x_i$, $\hat{\lambda}_{c,i}=\hat{\lambda}_{c',i}$. We order the buckets in increasing order of the volume of the hyperrectangles in it, denoted as $V[B_t]$.
%Because of the previous lemma, we know that for two hyperrectangles $\hat{\Lambda}_c\in B_t, \hat{\Lambda}_{c'}\in B_{t'}$ such that $i<i'$, we have that for all $i\in[k]$, $\hat{\lambda}_{i,x}\leq\hat{\lambda}_{i',x}$.

\begin{lemma}
    Let $B_t, B_{t'}$ be buckets with $t<t'$. 
    Then, $V[B_{t'}]/V[B_t]$ hyperrectangles from $B_t$ can be packed to form one hyperrectangle with the same shape as the hyperrectangles in $B_{t'}$.
\end{lemma}
\begin{proof}
    Let $\hat{\Lambda}_c\in B_t, \hat{\Lambda}_{c'}\in B_{t'}$. By~ \autoref{lemma:machineordering}, all dimensions of $\hat{\Lambda}_{c'}$ are at least as big the corresponding dimension of $\hat{\Lambda}_c$. More specifically, since the side lengths are of the form $2^{\alpha_{c,i}}$, we know that $\hat{\lambda}_{c',i}=2^{a_{i,t}} \cdot \hat{\lambda}_{c,i}$ for some $a_{i,t}\in\mathbb{Z}^+$.
    For some dimension $i$, take $2^{a_{i,t}}$ hyperrectangles from the bucket $B_t$ and stack them together in the dimension $i$. This will create one hyperrectangle where dimension $i$ is the same as dimension $i$ in $\hat{\Lambda}_{c'}$. We now continue this process across all the other dimensions. Let $a_t=\sum_{i\in[k]}a_{i,t}$. Then, this process uses $2^{a_t} = V[B_{t'}]/V[B_t]$ hyperrectangles of shape $\hat{\Lambda}_c$. Note that for at least one $i$, $a_{i,t}>0$, since otherwise $\hat{\Lambda}_c=\hat{\Lambda}_{c'}$, and then they are in the same bucket.
\end{proof}

The above lemma means that for each adjacent pair of buckets, $B_t, B_{t+1}$, if $B_t$ contains at least $V[B_{t+1}]/V[B_t]$ hyperrectangles, we can merge them into one hyperrectangle in $B_{t+1}$. The packing algorithm will merge as many hyperrectangles as possible, starting with the smallest bucket. When there are no merges left possible, each bucket $B_t$ has at most $V[B_{t+1}]/V[B_t]-1$ hyperrectangles, since otherwise another merge is possible. We can now show that by only using the rectangles in the largest bucket $B_b$, we can almost cover the whole output space.

%\begin{lemma}    \label{smallbucketscontainedlemma}
%    Assume no merges are possible for the sequence of bucket $B_1,...,B_b$. Then, $\sum_{i=1}^{b-1}V[B_t](2^{a_i-1})<V[B_{b}]$.
%\end{lemma}
%\begin{proof}
%    We prove this with induction. For $b=2$, we want to show $V[B_1](2^{a_1}-1)<V[B_2]$. This is true since $V[B_1]2^{a_1}=V[B_2]$.
%    Assume the equation is true for some $b=n$. Then:
%    \begin{align*}
%    \sum_{i=1}^{n}V[B_t](2^{a_i}-1)
%    =V[B_n](2^{a_n}-1)+\sum_{i=1}^{n-1}V[B_t](2^{n_{cj}}-1) \\
%    <V[B_n](2^{a_n}-1)+V[B_{n}]
%    =V[B_n]2^{a_n}
%    =V[B_{n+1}]
%    \end{align*}
%\end{proof}

\begin{lemma}
    \label{lemma:BbCover}
    Let $p_t$ be the number of hyperrectangles in bucket $B_t$, for $i\in \{1, \dots,b\}$. Then, $\vol{\Lambda} < (1+p_b)V[B_b]$.
\end{lemma}
\begin{proof}
    We use the observation that across all buckets the total volume is at least $\vol{\Lambda}$. Then:
    \begin{align*}
    \vol{\Lambda} & \leq \sum_{t=1}^bp_tV[B_t]
     = p_bV[B_b]+\sum_{t=1}^{b-1}p_tV[B_t] 
     \leq p_bV[B_b]+\sum_{t=1}^{b-1}\left(\frac{V[B_{t+1}]}{V[B_t]}-1\right)V[B_t] \\
    & =  p_bV[B_b]+\sum_{t=1}^{b-1}(V[B_{t+1}]-V[B_t]) 
    \leq (1+p_b)V[B_b] 
    \end{align*}
    where the second inequality holds because $p_t < V[B_{t+1}]/V[B_t]$ and the last inequality holds because it is is a telescopic sum.
\end{proof}

We will now pack $\Lambda$ using only the $p_b$ hyperrectangles in the last bucket $B_b$. Let $\hat{n}$ be the domain $n$ rounded to the nearest higher power of two. Note that no dimension of a hyperrectangle $\hat{\Lambda}_c\in B_b$ is greater than $\hat{n}$. This is because $\lambda_{c,i}\leq n$, so $\hat{\lambda}_{c,i}\leq \hat{n}$.
We will merge the hyperrectangles in $B_b$ the following way. Find the minimum dimension $\hat{\lambda}_{c,i}$ of $\hat{\Lambda}_c\in B_b$, and pairwise merge hyperrectangles in $B_b$ into hyperrectangles of volume $2V[B_b]$ by putting them adjacent in dimension $i$. This creates a new bucket $B_{b+1}$, with $\lfloor p_b/2\rfloor$ hyperrectangles, and at most one hyperrectangle in $B_b$ is left unmerged. This process can be repeated on hyperrectangles in $B_{b+1}$, until a bucket $B_{b+d}$ is obtained, where $B_{b+d}$ contains one hyperrectangle, so no further merges are possible. We will now show that we can cover $\Lambda$ using just the one hyperrectangle in $B_{b+d}$, by scaling it up by at most a constant factor.

\begin{lemma}
$V[B_{b+d}]> \vol{\Lambda}/{2}$.
\end{lemma}

\begin{proof}
Let $\beta$ be the number of hyperrectangles in $B_b$ after the first merge. Denote by $p_t$ the number of hyperrectangles in bucket $B_t$ after the last merge step for each $t \in \{b,\dots, b+d\}$, $p_t\in\{0,1\}$. Since $V[B_{t+1}]/V[B_t]=2$,
    $$
    \sum_{t=b}^{b+d-1}p_tV[B_t] 
    \leq \sum_{t=b}^{b+d-1} V[B_t] \leq V[B_{b+d}]-V[B_b]
    $$
Moreover, we have:
    $$
    \beta V[B_b]=\sum_{t=b}^{b+d} p_tV[B_t]=V[B_{b+d}]+\sum_{t=b}^{b+d-1}p_tV[B_t] 
        \leq 2V[B_{b+d}]-V[B_b]
    $$
    Finally, by reorganizing the above inequality and applying~\autoref{lemma:BbCover}, we obtain that $V[B_{b+d}] \geq (\beta+1)V[B_b]/2 > \vol{\Lambda}/2$.
\end{proof}

The above lemma shows that $R\in B_{b+d}$ almost covers $\Lambda$. There might however exist variables $x_i$ such that $|R_i| < n$. We will scale $R$ in each dimension $i$ by a factor $f_i=\max\{n/|R_i|,1\}$. This will guarantee that for each $x_i, i\in[k]$, $|R_i|\geq n$, and hence $R$ covers $\Lambda$. To scale $R$ by a factor $f_i$ in dimension $i$, we have to scale each $\hat{\Lambda}_c$ that is packed into $R$ by that same factor $f_i$ in dimension $i$, which gives the final sizes of hyperrectangles, which we denote $\Bar{\Lambda}_c$. If hyperrectangle $c$ is packed into $R$, $\Bar{\lambda}_{c,i}=f_i\hat{\lambda}_{c,i}$. If hyperrectangle $c$ is not packed into $R$, $\Bar{\lambda}_{c,i}=0$ since the hyperrectangle is not used.

\begin{lemma}
    \label{lemma:scaleBD}
    Let $R\in B_{b+d}$ be the remaining hyperrectangle. Scale $R$ in dimension $i$ by a factor $f_i=\max\{n/|R_i|,1\}$. Then $R$ covers $\Lambda$, and we have scaled $R$ in such a way that for each subset $S\subseteq[k]$, the following holds: $ \prod_{i\in S}f_i\leq 2^{k+1}$.
\end{lemma}
\begin{proof}
    The choice of $f_i=\max\{n/|R_i|,1\}$ means that $|R_i| f_i\geq n$. Hence $\Lambda$ is covered. We know that $\prod_{i\in[k]}R_i\geq n^k/2$, by the previous lemma. Note that hyperrectangles in $B_b, \dots B_{b+d}$, and hence also $R$, have side lengths at most $\hat{n}$ ($n$ rounded up to the nearest higher power of two) since we merged the smallest dimensions first. Furthermore, $\hat{n}<2n$. Therefore, for each $i\in[k]$, $|R_i|\leq 2n$. Now,
    \begin{align*}
    \prod_{i\in[k]}f_i
    &=\prod_{i\in[k]}\max\{n/|R_i|,1\}
    =\prod_{i\in[k]}\frac{\max\{n,|R_i|\}}{|R_i|}
    =\frac{\prod_{i\in[k]}\max\{n,|R_i|\}}{V[R]} \\
    &\leq\frac{\prod_{i\in[k]}2n}{V[\Lambda]/2} 
    =\frac{2\cdot2^kn^k}{n^k}=2^{k+1}
    \end{align*}
    For any $S\subseteq[k]$, the product $\prod_{i\in S}f_i$ would be less than the product above, since for all $i\in[k]$, $f_i\geq 1$.
\end{proof}

We can now prove the main theorem about packing.

\begin{proof}[Proof of~\autoref{theorem:packingtheorem}]
    Let $\Lambda_c$ be the hyperrectangle of machine $c$ as given by~\autoref{theorem:partitionspace} and let $\Bar{\Lambda}_c$ be the hyperrectangle after the packing algorithm has run. The packing algorithm can increase sides $\lambda_{c,i}$ first by rounding up to $\hat{\lambda}_{c,i}$, which is at most a factor $2$ bigger. If hyperrectangle $c$ is included in the final hyperrectangle $R$, sides of $\hat{\Lambda}_c$ might then be scaled up again by a factor $f_i=\min\{1,n/|R_i|\}$, to $\Bar{\lambda}_{c,i}$. For an atom $S_j$ with arity $r_j$, we now have:
    $$
    \frac{\vol{\pi_{S_j}\Bar{\Lambda}_c}}{\vol{\pi_{S_j}{\Lambda}_c}}
    =\prod_{x_i \in S_j}\frac{\hat{\lambda}_{c,i}}{\lambda_{c,i}} \cdot \frac{\Bar{\lambda}_{c,i}}{\hat{\lambda}_{c,i}}
    \leq \prod_{x_i \in S_j}2 \cdot f_i \leq 2^{r_j+k+1}
    $$
    The second inequality comes from \autoref{lemma:scaleBD}.
\end{proof}

\subsection{Putting Everything Together}

We can now prove the main theorems in this section.

\begin{proof}[Proof of~\autoref{theorem:UpperBoundDense}]
    The worst case load of the algorithm is that every possible tuple in $\Lambda_c$ exists. We will calculate the load of machine $c$ from relation $S_j$, $L_{cj}$. Denote $n_{cj}$ as the number of tuples received by machine $c$ from $S_j$. We get
    \[
    L_{cj}=\frac{n_{cj}\log n}{w_c}
    =\frac{\log n}{w_c} |\pi_{S_j}\Lambda_c|
    \leq \frac{1}{w_c}\frac{w_c}{\norm{\mathbf{w}}{v}}n^r\log n
    =O\left(\frac{n^r \log n}{\norm{\mathbf{w}}{v}}\right)
    \]
    Here the inequality comes from \autoref{theorem:partitionspace}. The result follows since the query has a constant number of atoms.
\end{proof}

\begin{proof}[Proof of~\autoref{theorem:UpperBoundSparse}]
    Denote $N_{cj}$ as the number of bits received by machine $c$ from relation $S_j$. The probability that a tuple $a_j\in S_j$ maps to machine $c$ is the following:
    \[
        Pr[(\pi_{S_j}\mathbf{h})(a_j)\in \Lambda_c]
        =\frac{|\pi_{S_j}\Lambda_c|}{n^{r_j}}
        \leq\frac{w_c}{\norm{\mathbf{w}}{v}n^r}n^r
        =\frac{w_c}{\norm{\mathbf{w}}{v}}
    \]
    The inequality comes from \autoref{theorem:partitionspace}. Note that since we use hashing and have a matching database instance, the probability that a tuple is mapped to machine $c$ is the same and independent among all tuples in the hyperrectangle. Therefore, $n_{cj} \sim Bin(m, \frac{w_c}{\norm{\mathbf{w}}{v}})$. We get the following expected value
    \[
    E[L_{cj}]
    =\frac{1}{w_c}E[n_{cj}]\log n
    =\frac{1}{w_c}\frac{w_cm}{\norm{\mathbf{w}}{v}}\log n
    =O\left(\frac{m\log n}{\norm{\mathbf{w}}{v}}\right)
    \]
    We also show that the probability that the load is more than this is exponentially small. Indeed, applying the Chernoff bound, which we describe in \autoref{appendix:Chernoff}, we have:
    \begin{align*}
        Pr\left[L_{cj}\geq(1+\delta)\frac{m\log n}{\norm{\mathbf{w}}{v}}\right]
        & =Pr\left[ N_{cj} \geq (1+\delta)\frac{w_cm\log n}{\norm{\mathbf{w}}{v}} \right] \\
        & =Pr\left[ n_{cj}\geq(1+\delta)\frac{mw_c}{\norm{\mathbf{w}}{v}} \right]
        \leq exp\left( -\delta^2\frac{mw_c}{3\norm{\mathbf{w}}{v}} \right)
    \end{align*}
We obtain the probability bound by taking the union bound across all atoms and machines.
\end{proof}

\section{Lower Bounds}
\label{sec:lower}

We present a lower bound on the load when machines have linear cost functions and all atoms have the same cardinality. This lower bound applies to both the sparse and the dense case, and considers the behavior of the algorithm over a probability distribution of inputs.
%In the dense distribution, we require the arity to be the same in every atom, since otherwise the cardinalities $\theta n^r$ differ.

We consider again for each machine $c$ a linear cost function $g_c(N)=N/w_c$ with $\mathbf{w}=(w_1,...,w_p)$ as weights. Let $\textbf{u}=(u_1,...,u_l)$ be a fractional edge packing for $q$, with $u=\sum_{j\in[l]}u_j$. Moreover, let $m$ be the cardinality of every relation. Then, define
$$ L^{\textsf{lower}}_{\bu} := \frac{m}{(\sum_{c\in[p]}w_c^u)^{1/u}}=\frac{m}{\norm{\mathbf{w}}{u}}$$

\begin{theorem}
    \label{lowerBoundEquiSizeLinear}
    Let $q$ be a CQ and let $\textbf{u}=(u_1,...,u_l)$ be a fractional edge packing for $q$. 
    Consider the uniform probability distribution $\mathcal{I}$ of matching databases with $m$ tuples per relation over domain $[n]$. Denote by $E_{I \sim \mathcal{I}}[|q(I)|]$ the expected value of the number of output tuples $|q(I)|$, over instances $I$ in the probability distribution $\mathcal{I}$. 
    Then, any one-round algorithm that in expectation outputs at least $E_{I \sim \mathcal{I}}[|q(I)|]$ tuples has load $\Omega(L^{\textsf{lower}}_{\bu})$ in the linear cost model. 
    The same lower bound holds for the probability distribution $\mathcal{I}^d$ of $\theta$-dense instances over domain $[n]$.
\end{theorem}

Each fractional edge packing $\textbf{u}$ gives a different lower bound, the highest of which is obtained by minimizing $\norm{\mathbf{w}}{u}$. Since the $p$-norm is a decreasing function of $p$, the highest lower bound is given by the maximum fractional edge packing. The maximum fractional edge packing is equal to the minimum vertex cover via duality of linear programs, hence the lower bound matches our upper bounds within a logarithmic factor. 

\begin{theorem}
    $\min_{\bv} (L^{\textsf{upper}}_\bv) = \log n \cdot \max_{\bu} (L^{\textsf{lower}}_\bu) $
\end{theorem}

%The following theorems state that the upper bounds stated in \autoref{theorem:UpperBoundDense} and \autoref{theorem:UpperBoundSparse} almost match the lower bounds (the proofs are given in \autoref{appendix:RelationEntropy})

%\begin{restatable}{theorem}{upperlowermatchdense}
%\label{theorem:LoadMatchDense}
%    Let $L$ be the load achieved by the upper bound algorithm for the dense data distribution. Let $L^*$ be the lower bound on the load for the dense distribution. Then $L=O(L^*\log n)$.
%\end{restatable}

%\begin{restatable}{theorem}{upperlowermatchsparse}
%\label{theorem:LoadMatchSparse}
%    Let $L$ be the load achieved by the upper bound algorithm for the sparse data distribution. Let $L^*$ be the lower bound on the load for the sparse distribution. Then $L=O(L^*\log n)$.
%\end{restatable}

We will next give an overview of the proof of~\autoref{lowerBoundEquiSizeLinear}, with some details left to \autoref{appendix:LowerBoundLemma} and \autoref{appendix:RelationEntropy}. We assume that initially each relation is stored at a separate location. Let $\text{msg}_j$ be the bit string that a fixed machine receives from $S_j$, and let $\text{msg}$ be the concatenation of $\text{msg}_j$ for all $j$. Note that $|\text{msg}|$ is the number of bits that the machine receives. We let $\text{Msg}(I)$ be the random variable mapping from the set of possible database instances to the value of $\text{msg}$. $\text{Msg}_j(S_j)$ is defined in the same way but maps to $\text{msg}_j$.

\begin{definition}
Let $R$ be a relation, and let $a\in R$ be a tuple. We say that $a$ is known by the machine, given message $\text{msg}$, if for all all database instances $I$ where $\text{Msg}(I)=\text{msg}$, $a\in R$. We denote the set of known tuples by machine $c$ given message $\text{msg}$ as $K_{msg}^c(R)$. Furthermore, we define $K_{msg}(R)=\bigcup_c K_{msg}^c(R)$.
\end{definition}

For each $S_j$, let $f_{c,j} \in [0,1]$ be the maximum length of the message $msg_j$ that $c$ receives (across all instances in the distribution) divided by $M_j$, the number of bits in the encoding of $S_j$. Note that since we use the optimal encoding, $M_j$ is the entropy of our input distribution.

\begin{restatable}{lemma}{EntropyOMJDense}
\label{lemma:EntropyOMJDense}
    In a $\theta$-dense database, $M_j=\Omega(m_j)$.
\end{restatable}
\begin{restatable}{lemma}{EntropyOMJSparse}
\label{lemma:EntropyOMJSparse}
    In a matching database, $M_j=\Omega(m_j)$.
\end{restatable}

We prove the above lemmas in \autoref{appendix:RelationEntropy}. 
To show \autoref{lowerBoundEquiSizeLinear}, we use the following lemma, which we prove in \autoref{appendix:LowerBoundLemma}. The lemma was proven in \cite{communication_steps_journal} for matching databases. However, some assumptions about the data distributions are changed and we also show the lemma for a $\theta$-dense database distribution. 

\begin{restatable}{lemma}{LowerBoundLemma}
    \label{lemma:LowerBoundLemma}
    Let $\mathbf{u}=(u_1,\dots,u_l)$ be a fractional edge packing of $q$. Then the expected number of known output tuples is
    $$ E[|K_{msg}^c(q(I))|] \leq \prod_{j\in[l]}f_{c,j}^{u_j} \cdot E[|q(I)|] $$
\end{restatable}

%However, the techniques are similar to those used in earlier lower bounds for CQs in the MPC model in one round \cite{communication_steps_journal, beame2014skewparallelqueryprocessing}.
We can now prove the main theorem of this section. We will use the notation $u=\sum_{j\in[l]}u_j$.

\begin{proof}[Proof of \autoref{lowerBoundEquiSizeLinear}]
From the definition of the load, $f_{c,j}\leq Lw_c/M$. 
Applying~\autoref{lemma:LowerBoundLemma},
    \[
    E[|K_{msg}^c(q(I))|]
    \leq E[|q(I)|]\cdot \prod_{j\in[l]}f_{c,j}^{u_j} 
    \leq E[|q(I)|]\cdot \prod_{j\in[l]}\left(\frac{Lw_c}{M}\right)^{u_j}
    =E[|q(I)|]\left(\frac{Lw_c}{M}\right)^u
    \]
We now use that $|K_{msg}(q(I))|=|\bigcup_{c\in[p]}K_{msg}^c(q(I))|\leq \sum_{c\in[p]}|K_{msg}^c(q(I))|$.
    \[
    E[|K_{msg}(q(I))|]
    \leq \sum_{c\in[p]} \left[\left(\frac{Lw_c}{M}\right)^uE[|q(I)|]\right]
    = \frac{L^u}{M^u}E[|q(I)|]\sum_{c\in[p]}w_c^u.
    \]
    If the algorithm is to produce the whole output of the query, the expected number of known output tuples has to be at least the expected output size of the query, so
    $$
    \frac{L^u}{M^u}E[|q(I)|]\sum_{c\in[p]}w_c^u\geq E[|q(I)|]
    $$
    Use \autoref{lemma:EntropyOMJDense} or \autoref{lemma:EntropyOMJSparse}. This concludes the proof.
\end{proof}

\section{General Cost Functions}
\label{sec:general}

In previous sections, we considered machines with linear cost functions. In this section, we extend the result to a broader class of cost functions, where each machine $c$ is equipped with a general cost function $g_c$.

\begin{definition}
    \label{def:CostFunction}
    A cost function $g: \mathbb{Z}^+ \rightarrow \mathbb{R}^+$ is {\em well-behaved} if it satisfies the following:
    \begin{enumerate}
        \item $g(0)=0$;
        \item $g$ is increasing;
        \item there exists a constant $a>1$ such that for all $x\geq 1$, $g((1+\delta)x)\leq (1+\delta)^{a} g(x)$
    \end{enumerate}
\end{definition}

These restrictions on a cost function are natural, since the cost of receiving zero bits should be zero, and the cost of receiving additional bits should be positive. The last condition states that a cost function cannot grow faster than some polynomial at each point. This requirement is not required in the lower bound, but without it, it is difficult to create a matching upper bound since just one bit in addition to what is expected can arbitrarily increase the cost of the machine.

\begin{definition}
    For a well-behaved cost function $g$, define the function $g^*:\mathbb{R}^+\rightarrow\mathbb{Z}^+$ by:
    $$
    g^*(L) :=\max_{x\in\mathbb{Z}^+} \{ g(x)\leq L \}
    $$
\end{definition}

Under the above definition, $g_c^*(L)$ can be interpreted as the maximum number of bits the cost function permits the machine $c$ to receive with load at most $L$. The restriction $g_c(0)=0$ implies that if $g_c^*(L)$ is defined for some $L$, it is also defined for all $L'\in\mathbb{R}^+$ where $L'<L$.

\subsection{Lower Bound}

Given a query $q$, consider any fractional edge packing $\bu=(u_1,...,u_l)$ with $u = \sum_{j \in [l]} u_j$. Suppose each relation has uniform cardinality $m$. Then, define $\bar{L}_\bu^{\textsf{lower}}$ to be the minimum $L \geq 0$ that satisfies the following inequality.
$$\sum_{c\in[p]}(g_c^*(L))^u\geq m^u $$

\begin{theorem}
    Let $q$ be a CQ and let $\textbf{u}=(u_1,...,u_l)$ be a fractional edge packing for $q$. 
    Consider the uniform probability distribution $\mathcal{I}$ of matching databases with $m$ tuples per relation over domain $[n]$. Denote by $E_{I \sim \mathcal{I}}[|q(I)|]$ the expected value of the number of output tuples $|q(I)|$, over instances $I$ in the probability distribution $\mathcal{I}$. 
    Then any one-round algorithm with well-behaved cost functions $\{g_c\}_c$ that in expectation outputs at least $E_{I \sim \mathcal{I}}[q(I)]$ tuples has load $\Omega(\bar{L}^{\textsf{lower}}_{\bu})$. The same lower bound holds for the probability distribution $\mathcal{I}^d$ of $\theta$-dense instances over domain $[n]$. 
\end{theorem}

\begin{proof}
    This proof is similar to the proof of \autoref{lowerBoundEquiSizeLinear}. We apply again \autoref{lemma:LowerBoundLemma} and sum over all machines.
    \begin{align*}
    E[|K_{msg}(q(I))|] \leq E[|q(I)|]\sum_{c\in[p]}\prod_{j\in[l]}f_{c,i}^{u_j}
    &\leq E[|q(I)|]\sum_{c\in[p]}\prod_{j\in[l]}\left(\frac{g_c^*(L)}{M}\right)^{u_j} \\
    &= E[|q(I)|]\sum_{c\in[p]}\left(\frac{g_c^*(L)}{M}\right)^{u}
    \end{align*}
    Here the second inequality comes from that $f_{c,i}M\leq g_c^*(L)$, since a machine can not receive more bits than what is permitted by the load. Use that $M=O(m)$ by \autoref{lemma:EntropyOMJDense} or \autoref{lemma:EntropyOMJSparse}. We require that $E[|K_{msg}(q(I))|]\geq E[|q(I)|]$. This proves the theorem.
\end{proof}

The highest lower bound is given by the $\mathbf{u}$ that maximizes $\bar{L}^{\textsf{lower}}_{\bu}$.  We will now prove that the maximum fractional edge packing $\mathbf{u^*}$ always gives the best lower bound. We will assume that $m \geq 1$, meaning the database is not empty.

\begin{lemma}
Let $\mathbf{u^*}$ be the maximum fractional edge packing. Then, $\bar{L}^{\textsf{lower}}_{\bu^*} = \max_\bu \bar{L}^{\textsf{lower}}_{\bu}$.
%The lower bound given by the maximal fractional edge packing $\mathbf{u^*}$ is always a maximal lower bound.
\end{lemma}
\begin{proof}
    Let $L^* = \bar{L}^{\textsf{lower}}_{\bu^*}$. Suppose $L'$ is another lower bound given by another edge packing $u'$ with $u' \leq u^*$. It suffices to show that $L^*$ satisfies $\sum_{c\in[p]}(g_c^*(L^*))^{u'}\geq m^{u'}$, since then the lowest $L$ satisfying the equation can be at most $L^*$.

    Note that for all $c$, $g_c^*(L^*)^{u^*-u'}\leq m^{u^*-u'}$, since $g_c^*(L^*)\leq m$. Then,
    \[
    \sum_{c\in[p]}(g_c^*(L^*))^{u'} \geq
    \sum_{c\in[p]}(g_c^*(L^*))^{u^*} m^{u'-u^*}
    = m^{u'-u^*} \sum_{c\in[p]}(g_c^*(L^*))^{u^*}
    \geq m^{u'-u^*} \cdot m^{u^*} = m^{u'}
    \]
    where the last inequality follows from the fact that $L^*$ satisfies $\sum_{c\in[p]}(g_c^*(L^*))^{u^*} \geq m^{u^*}$.
 %   We also have $s(u^*)\geq M^{u^*}$ from \autoref{lowerBoundEquiSizeGeneral}. By multiplying these inequalities, $s(u')\geq M^{u'}$.
\end{proof}

\begin{example}
As an example, consider cost functions of the form $g_c(x)=\frac{x^a}{w_c}$, where $a>0$. Then $g_c^*(L)=(Lw_c)^{1/a}$. The lower bound then becomes:
$$
L \geq  \max_{\bu}\max_{c\in[p]} \left( \frac{m^a}{(\sum_{c\in[p]}w_c^{u/a})^{a/u}} \right)
$$
\end{example}

\subsection{Upper Bound}

We give an algorithm for evaluating full CQs with equal cardinality atoms, where each cost function $\{g_c\}_c$ is well-behaved. The approach is similar to linear cost functions, but we will need another method to pick the dimensions of each hyperrectangle $\{\Lambda_c\}_c$. 

The algorithm will require the numerical value of $\bar{L}^{\textsf{lower}} = \max_{\bu} \bar{L}^{\textsf{lower}}_\bu$, the lower bound on the load. We therefore need a method to find $\bar{L}^{\textsf{lower}}_{\bu^*}$ for the maximal fractional edge packing $\bu^*$. For this, we need to find the minimal positive value of the function $f(L)=\sum_{c\in[p]}(g_c^*(L))^{u^*}-m^{u^*}$. We know that $L$ is more than $0$ and at most $L_{max}=\min_{c\in[p]} g_c(m)$, since the query can be computed with load $L_{max}$ with just one machine. $\bar{L}^{\textsf{lower}}_{\bu^*}$ can be found using binary search on this interval.

Given the value $L^* = \bar{L}^{\textsf{lower}}$, our algorithm computes the hyperrectangles using the same general technique as in Section~\ref{sec:upper}, with the difference that the sizes of each dimension are different. In particular, for the minimum vertex cover $\bv$, we calculate the $i$-dimension for machine $c$ as follows:
    $$ \lambda_{c,i}:=\left(\frac{g_c^*(L^*)}{m}\right)^{v_i}n $$
As we show in the appendix, the dimensions we choose are such that we can still apply the same packing technique as in  Section~\ref{sec:upper}. Thus:    

\begin{restatable}[Dense Inputs]{theorem}{UpperBoundDenseGeneral}
\label{theorem:UpperBoundDenseGeneral}
Let $q$ be a full CQ with uniform arity $r$ and a $\theta$-dense input $I$ with domain $[n]$ (every relation has size $m=\theta n^r$). Then, we can evaluate $q$ in one round with well-behaved cost functions $\{g_c\}_c$ with load $O(\bar{L}^{\textsf{lower}} \cdot \log n)$.
\end{restatable}

\begin{restatable}[Sparse Inputs]{theorem}{UpperBoundSparseGeneral}
\label{theorem:UpperBoundSparseGeneral}
Let $q$ be a full CQ over a matching instance $I$ with uniform cardinalities $m$ and domain $[n]$. Then, we can evaluate $q$ in one round with well-behaved cost functions $\{g_c\}_c$ with load $O(\bar{L}^{\textsf{lower}} \cdot \log n)$ with high probability.
\end{restatable}

%Here $\tilde{O}(L^*)$ means that logarithmic factors are hidden. 
The proofs of the above theorems are provided in \autoref{appendix:GeneralCostUpperBound}.

\section{Different Cardinality Relations}
\label{sec:unequal}

We now move on to the general case where we do not require the cardinality of every atom to be the same. We will assume linear cost functions, that is, cost functions have the form $g_c(N)=N/w_c$. We will give a general lower bound and matching upper bounds for the cartesian product, the binary join, the star query and the triangle query.

\subsection{Lower Bound}

An important difference in the lower bound we will present next, to the lower bound for queries of equal cardinalities, is that we need to consider different edge packings for each different machine. We will denote the edge packing for query $q$ and machine $c$ as $\textbf{u}_c=(u_{c,1},\dots,u_{c,l})$.

\begin{theorem}
    \label{theorem:DiffCardLowerBound}
    Let $q$ be a CQ and let $\textbf{u}_c=(u_{c,1},\dots,u_{c,l})$ be any fractional edge packing for $q$ and machine $c$.
    Consider the uniform probability distribution $\mathcal{I}$ of matching databases with $m_j$ tuples for relation $S_j$ over domain $n$.
    Denote by $E_{I \sim \mathcal{I}}[|q(I)|]$ the expected value of the number of output tuples $|q(I)|$, over instances $I$ in the probability distribution $\mathcal{I}$.
    Then, any one-round algorithm with linear cost functions that in expectation outputs at least $E_{I \sim \mathcal{I}}[q(I)]$ tuples has load $\Omega(L)$, where $L$ is the smallest load that satisfies the following equation
    \begin{equation}
        \label{equation:DiffCardLowerBound}
        \sum_{c\in[p]}\prod_{j\in[l]}\left( \frac{ Lw_c}{m_j} \right)^{u_{c,j}} \geq 1
    \end{equation}
    The same lower bound holds for the probability distribution $\mathcal{I}^d$ of $\theta$-dense instances over domain $[n]$. 
\end{theorem}

\begin{proof}
    The proof is similar to proofs of previous lower bounds. We will use \autoref{lemma:LowerBoundLemma} to bound the number of output tuples produced by one machine. We require that all machines together produce at least $E[|q(I)|]$ output tuples.
    \[
    E[|K_{msg}^c[q(I)]|] \leq \prod_{j\in[l]}f_{c,i}^{u_{c,j}}E[|q(I)|]
    \leq E[|q(I)|] \prod_{j\in[l]}\left( \frac{f_{c,i}M_j}{M_j} \right)^{u_{c,j}}
    \leq E[|q(I)|] \prod_{j\in[l]}\left( \frac{ Lw_c}{M_j} \right)^{u_{c,j}}
    \]
    Here we used that $f_{c,i}M_j\leq Lw_c$. This is because $Lw_c$ is the maximum number of bits machine $c$ can receive about each relation, to keep the load $L$. The theorem follows by taking a sum across all machines.
\end{proof}

The highest lower bound $L^*$ is given by the set of edge packings $\{\bu_c \}_{c\in[p]}$, one for each machine, that maximizes the load needed to satisfy the equation. Next, we show how to compute the numerical value of $L^*$, which will be used by the upper bound.

\begin{lemma}
    Let $L^*$ denote the maximum lower bound on the load from \autoref{theorem:DiffCardLowerBound}. Then:
    $$
    \frac{\max_j m_j}{\sum_{c\in[p]}w_c} \leq L^* \leq \frac{\max_j m_j}{\max_{c\in[p]}w_c}
    $$
\end{lemma}
\begin{proof}
    We start with the first inequality. Note that the edge packing where edge $j^*$ with maximum cardinality gets $u_{j^*}=1$, and all other edges get weight 0 is a valid edge packing. Hence $\sum_{c\in[p]}\frac{Lw_c}{m_{j^*}}\geq 1$ is a lower bound.

    For the second inequality, if only the biggest machine is used in computation, the load is $\max_j m_j/\max_{c\in[p]}w_c$. Therefore this load can always be achieved by computing the query on only the biggest machine. Since $L^*$ is the optimal load, it will never be more than this.
\end{proof}

Note that $p\cdot\max_{c\in[p]}w_c\geq\sum_{c\in[p]}w_c$. This together with the lemma above shows that the range of possible values for $L^*$ is at most a factor
$$
\frac{\max_jm_j/\max_{c\in[p]}}{\max_jm_j/\sum_{c\in[p]}w_c}
=\frac{\sum_{c\in[p]}w_c}{\max_{c\in[p]}w_c}
\leq \frac{p\cdot\max_{c\in[p]}w_c}{\max_{c\in[p]}w_c}
=p
$$

Since the range of possible values of $L^*$ is $p$, we can find $L^*$ by starting with a guess $\hat{L} \gets \frac{\max_jm_j}{\sum_{c\in[p]}w_c}$. We can check if our guess is correct by finding the edge packings $\{\bu_c\}_{c\in[p]}$ for each machine. We can then check if $\sum_{c\in[p]}\prod_{j\in[l]}\left(\frac{\hat{L}w_c}{m_j}\right)^{u_{c,j}}$ is at least $1$. If this is not the case, we double our guess $\hat{L}$. Since the range of possible values of $L^*$ is just a factor $p$, we have to iterate this procedure at most $\log p$ times.

%For the optimal edge packing $\bu^*$ and $L=L^*$, \autoref{equation:DiffCardLowerBound} becomes an equality. If it were not an equality, $L^*$ could have been slightly lowered, which would have contradicted that $L^*$ is the minimal $L$ that satisfies the equation. \paris{some more things}. Thus: 
%
%$$ \min_{\bu_1, \dots, \bu_p} \sum_{c \in [p]} \prod_{j \in [l]}\left( \frac{ L^* w_c}{M_j} \right)^{u_{c,j}} = 1$$
%
%Since minimizing each summand is independent of the others, we can equivalently write the above equation as follows:
%
%$$ \sum_{c \in [p]} \min_{\bu_i} \prod_{j\in[l]}\left( \frac{ L^* w_c}{M_j} \right)^{u_{c,j}} = 1$$
%
%Given the above equation, we can define
%
%$$ \alpha(w) :=\min_{\bu} \prod_{j\in[l]}\left(\frac{L^* w}{M_j}\right)^{u_j}$$
%
%for some weight $w$ of a machine. Using this notation, we have that $\sum_{c\in[p]}\alpha(w_c)=1$. Hence, we can view $\alpha(w_c)$ as the fraction of the volume of $\Lambda$ that is assigned to machine $c$ with weight $w_c$. Note that $\alpha(w)$ is a continuous function for $w > 0$ since it is a minimization of continuous functions.

%Also note that $L^*$ satisfies \autoref{equation:DiffCardLowerBound} for every edge packing, since it is given by the edge packing that maximizes what it has to be to satisfy the equation. We can therefore think of finding the edge packing of a machine as minimizing $\prod_{j\in[l]}(L^*w_c/M_j)^{u_{c,j}}$ under edge packing constraints.

%Note that if we know the function $\alpha(w)$, we can view the lower bound as follows:
%
%\[ \prod_{j\in[l]\}\left(\frac{L w}{M_j}\right)^{u_j} \geq \alpha(w) \]

\subsection{Upper Bound}

We will now show how to match the lower bound for the cartesian product, the binary join, the star query and the triangle query. 
For the general full CQ, creating an algorithm is challenging remains an open problem. 
The difficulty is that the lower bound in \autoref{theorem:DiffCardLowerBound} might give a different edge packing $\bu_c$ to each machine. 
A linear program to find the edge packing for a machine can be obtained by minimizing the logarithm of $\prod_{j\in[l]}(L^*w_c/m_j)^{u_{c,j}}$. 
By considering the dual program, it is possible to find hyperrectangles $\Lambda_c$ such that the expected load of machines matches the lower bound and the total volume of all $\Lambda_c$ cover $\Lambda$, similar to what was done in \cite{communication_steps_journal} for homogenous machines. 
However, for the packing to work, we need all sides of the hyperrectangle $\Lambda_c$ to increase when the weight of a machine increases. 
It is not clear how to guarantee this, or whether it is possible.
We have however been able to match the lower bound for specific queries, using the same algorithm as in previous sections, by modifying how shapes of subspaces are picked.
Here, we describe how to do this for the Cartesian Product and Binary Joins.
In \autoref{appendix:UnEqualCardinality}, we also show the Star Query and the Triangle Query.
Proof of correctness of the following algorithm is provided in \autoref{appendix:UnEqualCardinality}.

\paragraph*{Cartesian Product}

We consider the cartesian product $q(x,y) \dl S_1(x),S_2(y)$.

Let $L^*$ be the lower bound on the load. Let the side length of $\Lambda_c$ along dimension $j$ be
\[ 
\lambda_{c,i}:=\min\left(\frac{L^*w_c}{M_j}, 1\right)n = \left(\frac{L^*w_c}{M_j}\right)^{u_{c,i}}
\]

\paragraph*{Binary Join}

Next, consider the binary join.
\[ q(x,y,z)\dl S_1(x,z),S_2(y,z) \]
Without loss of generality, assume $|S_1|\geq|S_2|$.

Let $L^*$ be the lower bound on the load. Let the side lengths of $\Lambda_c$ be the following:
\[ 
\lambda_{c,x}:=n, \quad \lambda_{c,y}:=n, \quad
\lambda_{c,z}:=\frac{L^*w_c}{M_1}n 
\]

\paragraph*{Matching the Lower Bound}

The theorems below, which are proven in \autoref{appendix:UnEqualCardinality}, show that the algorithm that follows matches the lower bound. 

\begin{restatable}[Dense]{theorem}{UpperBoundDiffCardinalityDense}
    \label{theorem:UpperBoundDiffCardinalityDense}
    Let $q$ be one of the cartesian product, the binary join, the star query and the triangle query over a $\theta$-dense input $I$, where arities $r_j$ of tables do not have to be uniform. Let $L^*$ be the lower bound on the load. Then we can evaluate $q$ with heterogenous machines with weights $w_1,\dots,w_p$ with load $O(L^*\log n)$.
\end{restatable}

\begin{restatable}[Sparse]{theorem}{UpperBoundDiffCardinalitySparse}
    \label{theorem:UpperBoundDiffCardinalitySparse}
    Let $q$ be one of the cartesian product, the binary join, the star query and the triangle query over a matching database $I$ where atom $S_j$ has cardinality $m_j$. Let $L^*$ be the lower bound on the load. Then we can evaluate $q$ with heterogenous machines with weights $w_1,\dots,w_p$ with load $O\left(L^*\log n\right) $
    with high probability.
\end{restatable}

\section{Conclusion}

In this paper, we studied the problem of computing full Conjunctive Queries in parallel on heterogeneous machines. Our algorithms are inspired by the HyperCube algorithm but take a new approach of considering how to optimally partition the space of possible output tuples among machines. This gives an optimal algorithm for queries where relations have the same cardinalities, for both linear and more general cost functions, and an optimal algorithm for queries with atoms of any cardinality for specific queries.

\bibliography{references}

\begin{thebibliography}{10}

\bibitem{afratiullman}
Foto~N. Afrati and Jeffrey~D. Ullman.
\newblock Optimizing joins in a map-reduce environment.
\newblock In Ioana Manolescu, Stefano Spaccapietra, Jens Teubner, Masaru Kitsuregawa, Alain L{\'{e}}ger, Felix Naumann, Anastasia Ailamaki, and Fatma {\"{O}}zcan, editors, {\em {EDBT} 2010, 13th International Conference on Extending Database Technology, Lausanne, Switzerland, March 22-26, 2010, Proceedings}, volume 426 of {\em {ACM} International Conference Proceeding Series}, pages 99--110. {ACM}, 2010.
\newblock \href {https://doi.org/10.1145/1739041.1739056} {\path{doi:10.1145/1739041.1739056}}.

\bibitem{communication_steps_journal}
Paul Beame, Paraschos Koutris, and Dan Suciu.
\newblock Communication steps for parallel query processing.
\newblock {\em J. {ACM}}, 64(6):40:1--40:58, 2017.
\newblock \href {https://doi.org/10.1145/3125644} {\path{doi:10.1145/3125644}}.

\bibitem{topologyawaredataprocessing}
Spyros Blanas, Paraschos Koutris, and Anastasios Sidiropoulos.
\newblock Topology-aware parallel data processing: Models, algorithms and systems at scale.
\newblock In {\em 10th Conference on Innovative Data Systems Research, {CIDR} 2020, Amsterdam, The Netherlands, January 12-15, 2020, Online Proceedings}. www.cidrdb.org, 2020.
\newblock URL: \url{http://cidrdb.org/cidr2020/papers/p10-blanas-cidr20.pdf}.

\bibitem{snowflake}
Beno{\^{\i}}t Dageville, Thierry Cruanes, Marcin Zukowski, Vadim Antonov, Artin Avanes, Jon Bock, Jonathan Claybaugh, Daniel Engovatov, Martin Hentschel, Jiansheng Huang, Allison~W. Lee, Ashish Motivala, Abdul~Q. Munir, Steven Pelley, Peter Povinec, Greg Rahn, Spyridon Triantafyllis, and Philipp Unterbrunner.
\newblock The snowflake elastic data warehouse.
\newblock In Fatma {\"{O}}zcan, Georgia Koutrika, and Sam Madden, editors, {\em Proceedings of the 2016 International Conference on Management of Data, {SIGMOD} Conference 2016, San Francisco, CA, USA, June 26 - July 01, 2016}, pages 215--226. {ACM}, 2016.
\newblock \href {https://doi.org/10.1145/2882903.2903741} {\path{doi:10.1145/2882903.2903741}}.

\bibitem{friedgut}
Ehud Friedgut.
\newblock Hypergraphs, entropy, and inequalities.
\newblock {\em Am. Math. Mon.}, 111(9):749--760, 2004.
\newblock URL: \url{http://www.jstor.org/stable/4145187}.

\bibitem{coverorpack}
Xiao Hu.
\newblock Cover or pack: New upper and lower bounds for massively parallel joins.
\newblock In Leonid Libkin, Reinhard Pichler, and Paolo Guagliardo, editors, {\em PODS'21: Proceedings of the 40th {ACM} {SIGMOD-SIGACT-SIGAI} Symposium on Principles of Database Systems, Virtual Event, China, June 20-25, 2021}, pages 181--198. {ACM}, 2021.
\newblock \href {https://doi.org/10.1145/3452021.3458319} {\path{doi:10.1145/3452021.3458319}}.

\bibitem{topologyawarejoins}
Xiao Hu and Paraschos Koutris.
\newblock Topology-aware parallel joins.
\newblock {\em Proc. {ACM} Manag. Data}, 2(2):97, 2024.
\newblock \href {https://doi.org/10.1145/3651598} {\path{doi:10.1145/3651598}}.

\bibitem{hu2020algorithms}
Xiao Hu, Paraschos Koutris, and Spyros Blanas.
\newblock Algorithms for a topology-aware massively parallel computation model.
\newblock In Leonid Libkin, Reinhard Pichler, and Paolo Guagliardo, editors, {\em PODS'21: Proceedings of the 40th {ACM} {SIGMOD-SIGACT-SIGAI} Symposium on Principles of Database Systems, Virtual Event, China, June 20-25, 2021}, pages 199--214. {ACM}, 2021.
\newblock \href {https://doi.org/10.1145/3452021.3458318} {\path{doi:10.1145/3452021.3458318}}.

\bibitem{parallelAcyclicJoins}
Xiao Hu and Yufei Tao.
\newblock Parallel acyclic joins: Optimal algorithms and cyclicity separation.
\newblock {\em J. {ACM}}, 71(1):6:1--6:44, 2024.
\newblock \href {https://doi.org/10.1145/3633512} {\path{doi:10.1145/3633512}}.

\bibitem{Hu_2019_OutputSensitive}
Xiao Hu and Ke~Yi.
\newblock Instance and output optimal parallel algorithms for acyclic joins.
\newblock In Dan Suciu, Sebastian Skritek, and Christoph Koch, editors, {\em Proceedings of the 38th {ACM} {SIGMOD-SIGACT-SIGAI} Symposium on Principles of Database Systems, {PODS} 2019, Amsterdam, The Netherlands, June 30 - July 5, 2019}, pages 450--463. {ACM}, 2019.
\newblock \href {https://doi.org/10.1145/3294052.3319698} {\path{doi:10.1145/3294052.3319698}}.

\bibitem{ketzman}
Bas Ketsman, Dan Suciu, and Yufei Tao.
\newblock A near-optimal parallel algorithm for joining binary relations.
\newblock {\em Log. Methods Comput. Sci.}, 18(2), 2022.
\newblock URL: \url{https://doi.org/10.46298/lmcs-18(2:6)2022}, \href {https://doi.org/10.46298/LMCS-18(2:6)2022} {\path{doi:10.46298/LMCS-18(2:6)2022}}.

\bibitem{beame2016worstcaseoptimalalgorithmsparallel}
Paraschos Koutris, Paul Beame, and Dan Suciu.
\newblock Worst-case optimal algorithms for parallel query processing.
\newblock In Wim Martens and Thomas Zeume, editors, {\em 19th International Conference on Database Theory, {ICDT} 2016, Bordeaux, France, March 15-18, 2016}, volume~48 of {\em LIPIcs}, pages 8:1--8:18. Schloss Dagstuhl - Leibniz-Zentrum f{\"{u}}r Informatik, 2016.
\newblock URL: \url{https://doi.org/10.4230/LIPIcs.ICDT.2016.8}, \href {https://doi.org/10.4230/LIPICS.ICDT.2016.8} {\path{doi:10.4230/LIPICS.ICDT.2016.8}}.

\bibitem{Tao2020}
Yufei Tao.
\newblock A simple parallel algorithm for natural joins on binary relations.
\newblock In Carsten Lutz and Jean~Christoph Jung, editors, {\em 23rd International Conference on Database Theory, {ICDT} 2020, March 30-April 2, 2020, Copenhagen, Denmark}, volume 155 of {\em LIPIcs}, pages 25:1--25:18. Schloss Dagstuhl - Leibniz-Zentrum f{\"{u}}r Informatik, 2020.
\newblock URL: \url{https://doi.org/10.4230/LIPIcs.ICDT.2020.25}, \href {https://doi.org/10.4230/LIPICS.ICDT.2020.25} {\path{doi:10.4230/LIPICS.ICDT.2020.25}}.

\bibitem{amazonaurora}
Alexandre Verbitski, Anurag Gupta, Debanjan Saha, Murali Brahmadesam, Kamal Gupta, Raman Mittal, Sailesh Krishnamurthy, Sandor Maurice, Tengiz Kharatishvili, and Xiaofeng Bao.
\newblock Amazon aurora: Design considerations for high throughput cloud-native relational databases.
\newblock In Semih Salihoglu, Wenchao Zhou, Rada Chirkova, Jun Yang, and Dan Suciu, editors, {\em Proceedings of the 2017 {ACM} International Conference on Management of Data, {SIGMOD} Conference 2017, Chicago, IL, USA, May 14-19, 2017}, pages 1041--1052. {ACM}, 2017.
\newblock \href {https://doi.org/10.1145/3035918.3056101} {\path{doi:10.1145/3035918.3056101}}.

\end{thebibliography}

\newpage
\appendix

\section{Friedgut's Inequality}
\label{appendix:Friedgut}

We here give an inequality that will prove to be useful, that was first introduced by Friedgut \cite{friedgut}. Let $q$ be a Conjunctive Query. Let $\mathbf{a}\in [n]^k$ be a possible output tuple, and let $a_j\in[n]^{r_j}$ be the projection of $\mathbf{a}$ on the variables of atom $S_j$. For each possible tuple in each atom, denoted $a_j$, create a variable $p_j(a_j)\geq 0$. Let $u=(u_1,...,u_l)$ be a fractional edge cover over q. Then the following inequality holds.
\[
\sum_{\mathbf{a}\in [n]^k}\prod_{j\in[l]}p_j(a_j)
\leq \prod_{j\in[l]}\left( \sum_{a_j\in [n]^{r_j}} p_j(a_j)^{1/u_j} \right)^{u_j}
\]

\section{Chernoff Bound}
\label{appendix:Chernoff}

In proofs of upper bounds, we will use a Chernoff bound, a well known upper bound on tail probabilities. Let $X$ be a random variable following a binomial distribution $Bin(n,p)$. We wish to bound the probability of the event that $X$ deviates from the expected value by more than some constant factor $(1+\delta)$, i.e. the probability that $X \geq (1+\delta)E[X]$. Recall that $E[X]=np$ for binomial distributions. Then, the Chernoff bound states
\[
    Pr[X \geq (1+\delta)E[X]] \leq \exp\left(-\frac{\delta^2E[X]}{3}\right)
\]

\section{Entropy of Relations}
\label{appendix:RelationEntropy}

We give proofs for the following lemmas

\EntropyOMJDense*
\begin{proof}
    We can show this by bounding the entropy of the relation. Recall that $\binom{a}{b}=\binom{a}{a-b}$. Therfore, $H(S_j)$ is the same for a $\theta$-dense and a $(1-\theta)$-dense distribution. Let $\alpha=\max\{\theta,1-\theta\}$.
    \begin{align*}
    M_j &= H(S_j)=\log\left(\binom{n^{r_j}}{\theta n^{r_j}}\right)
    =\log\left(\binom{n^{r_j}}{\alpha n^{r_j}}\right)
    = \log\left( \frac{n^{r_j}!}{(\alpha n^{r_j})!((1-\alpha)n^{r_j})!} \right) \\
    &\geq \log\left( \frac{(\alpha n^{r_j})^{(1-\alpha)n^{r_j}}(\alpha n^{r_j})!}{(\alpha n^{r_j})!((1-\alpha)n^{r_j})!} \right)
    = \log\left(\frac{(\alpha n^{r_j})^{(1-\alpha)n^{r_j}}}{((1-\alpha)n^{r_j})!}\right) \\
    &\geq \log\left(\frac{(\alpha n^{r_j})^{(1-\alpha)n^{r_j}-1}}{((1-\alpha)n^{r_j}-1)!}\right)
    \geq \log\left(\left( \frac{\alpha n^{r_j}}{(1-\alpha)n^{r_j}-1)} \right)^{(1-\alpha)n^{r_j}-1}\right) \\
    &= \Omega((1-\alpha)n^{r_j})=\Omega(m_j)
    \end{align*}
    Here the first inequality comes from setting the first $(1-\alpha)n^{r_j}$ factors in $n^{r_j}!$ to $\alpha n^{r_j}$. The second inequality comes from removing a factor $\alpha n^{r_j}/((1-\alpha)n^{r_j})$ from the quotient. This is required to guarantee that the quotient in the logarithm after the last inequality is not equal to $1$, even when $\alpha = 1-\alpha$.
\end{proof}

\EntropyOMJSparse*
\begin{proof}
We will bound the entropy of the relation. Recall that for the sparse distribution, if the arity is $1$, we imposed that $m_j/n\leq\theta$ for some constant $\theta\in(0,1)$. We will prove the lemma in two cases. Begin with the case where the arity is $1$. We will use that $m_j/n\leq\theta$.
$$
M_j=H(S_j)=(r-1)\log(m_j!)+\log\binom{n}{m_j}^{r_j}
\geq \log\binom{n}{m_j}
$$
By the same argument as for the $\theta$-dense distribution, this is $\Omega(m_j)$. Next, we handle the case where the arity is more than $1$. The minimal entropy is when the arity is $2$.
$$
M_j=(r-1)\log(m_j!)+\log\binom{n}{m_j}^{r_j}
\geq \log(m_j!)+\log\binom{n}{m_j}^2 \geq \log(m_j!) = \Omega(m_j)
$$
\end{proof}

The naive encoding of a relation is to store all tuples that are present. For a relation $S_j$ of arity $r_j$, cardinality $m_j$ and constants from domain $[n]$, this encoding requires $r_jm_j\log n$ bits. The above lemmas indicate that for both data distributions, there might be cases where it is possible to encode a relation in $O(m_j)$ bits. In the dense distribution, one such $O(m_j)$ encoding is to store a bitmap over all possible $n^{r_j}$ tuples. For the sparse distribution, a $O(m_j)$ encoding exist in some cases. For example, if a matching relation $S_j$ has cardinality $m_j=n/2$ and arity $1$, it can also be stored with a bitmap over the $n$ possible tuples, which takes $O(m_j)$ bits.

\section{Lower Bound Lemma}
\label{appendix:LowerBoundLemma}

We will prove the following main lemma. The proof and the techniques used are similar to in \cite{communication_steps_journal}.

\LowerBoundLemma*

The overarching idea behind this lower bound is to use that $E[|K_{msg}(q(I))|]$, the expected number of known output tuples, has to be at least $E[q(I)]$, the expected number of output tuples, if the algorithm is to work. We can bound $E[|K_{msg}(q(I))|]$ if we know the expected number of known tuples of each atom, by using Fridguts inequality \autoref{appendix:Friedgut}.

\subsection{Bounding Entropy Reduction}

Let $S_j$ be any atom. We can view $S_j$ as a random variable, which takes values from the set of possible relation instances. The entropy of $S_j$, $H(S_j)$ denotes the number of bits needed to encode the relation $S_j$. We consider uniform distributions, meaning $S_j$ is uniformly distributed over all possible relation instances. Therefore $H(S_j)=\log_2(supp(S_j))$, where $supp(S_j)$ is the number of possible relation instances. When a machine receives information through communicating with other machines, the entropy is reduced because a subset of possible relation instances can be ruled out.

\begin{definition}
    For each relation $S_j$, define the entropy function $\psi:[n]\rightarrow \mathbb{R}$ as the function that maps from the number of unknown tuples in $S_j$ to the maximum remaining entropy. More formally, $\psi(x)=\max_{msg_j}H(S_j|\text{Msg}_j(I)=\text{msg}_j)$, where $msg_j$ is some message that reveals $m_j-x$ tuples of $S_j$. Note that each data distribution gives different entropy function.
\end{definition}
 
\begin{example}
    Take relations with a $\theta$-dense distributions, with attribute domain $[n]$, arity $r$, cardinality $\theta n^{r}$. Each tuple appears at most once with the same probability. The relation has the following entropy function.
    \[ 
    \psi_1(x)
    = \log \left( \binom{(1-\theta)n^{r}+x}{x} \right)
    \]
\end{example} 

\begin{example}
    For matching relations with attribute domain $[n]$, arity $r$ and cardinality $m$, we have the following entropy function.
    \[ \psi_2(x)=\log((x!)^{r-1}) + \binom{n-m+x}{x} \]
\end{example}

We will show that, for both the dense and the sparse data distribution, we can bound the entropy reduction of knowing some number of tuples in the relation using the number of received bits of information about the relation. More formally, we will show the following property for $\psi$, for some $\gamma=O(1)$, where $\alpha$ is the fraction of known tuples.

\begin{equation}
    \label{entropy_property}
    \psi((1-\alpha) m_j) \leq \left(1-\frac{\alpha}{\gamma}\right) \psi(m_j)
\end{equation}

This property means that if the number of tuples left to discover is reduced by a fraction $\alpha$, the entropy reduction is at most a constant $1/\gamma$ times the same fraction $\alpha$. If this is not true, machines can communicate some fraction $\alpha$ of the entropy of the relation, and learn about more than a constant times $\alpha$ fraction of the tuples in the relation. 

\begin{lemma}
\label{lemma:GammaLemma}
    Let $a,b\in \mathbb{R}^+$, where $b\leq a$, and let $\theta=b/a$. Then
    $$ \log\binom{a}{b}\leq \gamma b\log\left(\frac{a}{b}\right)$$
    where $\gamma=\left( 1 + \frac{(1-\theta)\log(1-\theta)}{\theta\log(\theta)} \right)$.
\end{lemma}
\begin{proof}
    We will use the following inequality.
    \[\log\binom{a}{b}\leq aH_2(b/a)=a(-(b/a)\log (b/a)-(1-b/a)\log(1-b/a))\]
    This inequality is well known and can be shown using the binomial theorem like below, and then using the logarithm.
    \[ 1=\left(\left(1-\frac{b}{a}\right)+\frac{b}{a}\right)^a\geq\binom{a}{b}\left(\frac{b}{a}\right)^b\left(1-\frac{b}{a}\right)^{a-b} \]
    We will now prove the lemma. Use that $\theta=b/a$
    \begin{align*}
    \log\binom{a}{b}\leq aH_2(\theta)
    =&a(-\theta\log\theta-(1-\theta)\log(1-\theta)) 
    =\frac{b}{\theta}(-\theta\log\theta-(1-\theta)\log(1-\theta)) \\
    =&-\left( \frac{\theta + \frac{(1-\theta)\log(1-\theta)}{log(\theta)}}{\theta} \right) b\log(\theta)
    =\left( 1 + \frac{(1-\theta)\log(1-\theta)}{\theta\log(\theta)} \right)b\log\left(\frac{a}{b}\right)
    \end{align*}
\end{proof}

We now prove that \autoref{entropy_property} holds for $\psi_1$ and $\psi_2$, both with $\gamma=\left( 1 + \frac{(1-\theta)\log(1-\theta)}{\theta\log(\theta)} \right)$.

\begin{theorem}
    For $\psi_1(x)=\log \left( \binom{(1-\theta)n^{r}+x}{x} \right)$, we have that $\psi_1((1-\alpha)m_j)\leq \left(1-\frac{\alpha}{\gamma}\right)\psi_1(m_j)$, where $\gamma=O(1)$.
\end{theorem}
\begin{proof}
    Note that $m_j=\theta n^r$. Define $k=\alpha m_j$.
    \begin{align*}
    \psi_1((1-\alpha)m_j)
    =&\log\left(\binom{(1-\theta)n^r+(1-\alpha)m_j}{(1-\alpha)m_j}\right)
    =\log\left(\binom{(1-\theta)n^r+(1-\alpha)\theta n^r}{(1-\alpha)m_j}\right)\\
    =&\log\left(\binom{n^r-\alpha \theta n^r}{\theta n^r - \alpha m_j}\right)
    =\log\left(\binom{n^r-k}{m_j-k}\right)
    \end{align*}
    It was proven in \cite{communication_steps_journal} that \autoref{entropy_property} is true for the above form of $\psi_1$ when $\theta\in(0,1/2]$. We will generalize this to $\theta\in(0,1)$, by using similar techniques. We have the following inequality.
    \[
    \frac{\binom{n^r-k}{m_j-k}}{\binom{n^r}{m_j}}
    =\frac{ \frac{(n^r-k)!}{(m_j-k)!(n^r-m_j)!} }{ \frac{n^r!}{m_j!(n^r-m_j)!}}
    =\frac{(n^r-k)!m_j!}{(m_j-k)!n^r!}
    =\frac{m_j(m_j-1)...(m_j-k+1)}{n^r(n^r-1)...(n^r-k+1)}
    \leq \left(\frac{m_j}{n^r}\right)^k
    \]
    We have the following.
    \[
    \psi_1((1-\alpha)m_j)
    \leq \log\left[ \left(\frac{m_j}{n^r}\right)^k \binom{n^r}{m_j} \right]
    = \log\binom{n^r}{m_j} - k\log\left(\frac{n^r}{m_j}\right)
    = \left( 1 - \frac{k\log\left(\frac{n^r}{m_j}\right)}{\log\binom{n^r}{m_j}} \right)\psi_1(m_j)
    \]
    We now use \autoref{lemma:GammaLemma}.
    \[
    \psi_1((1-\alpha)m_j)
    \leq \left(1-\frac{k\log\left(\frac{n^r}{m_j}\right)}{\gamma m_j\log\left(\frac{n^r}{m_j}\right)} \right)\psi_1(m_j)
    =\left(1-\frac{k}{\gamma m_j}\right)\psi_1(m_j)
    =\left(1-\frac{\alpha}{\gamma}\right)\psi_1(m_j)
    \]
\end{proof}

\begin{figure}[H]
    \caption{$\gamma$ as a function of $\theta$.}
    \begin{tikzpicture}
        \begin{axis}[ 
        xlabel=$\theta$,
        ylabel={$\gamma(\theta)$}
        ] 
        \addplot [
        domain=0.001:0.999, 
        samples=1000, 
        color=red
        ] {1+((1-x)*ln(1-x))/(x*ln(x)}; 
        \end{axis}
    \end{tikzpicture}
    \label{fig:gammaPlot}
\end{figure}
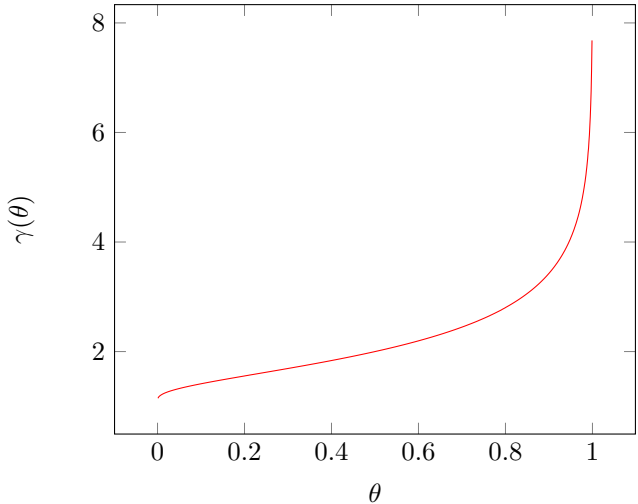

In \autoref{fig:gammaPlot}, we can see the value of the constant $\gamma$ as a function of $\theta$. When the data distribution gets denser, $\gamma$ increases. For $\theta=1/2$, $\gamma=2$. When $\theta$ approaches 1, $\gamma \rightarrow \infty$.

\begin{theorem}
    For $\psi_2(x)=\log((x!)^{r-1})+\log\binom{n-m_j+x}{x}$, we have that $\psi_2((1-\alpha)m_j)\leq (1-\frac{\alpha}{\gamma})\psi_2(m_j)$ for some $\gamma=O(1)$.
\end{theorem}
\begin{proof}
    Define $k=m-x$. The second term of the entropy function can be written as $\log\binom{n-k}{m-k}$. Take $\theta=m_j/n$, and then we can bound the term in the same way as for the dense distribution, for the same constant $\gamma$.
    For the other term:
    \begin{align*}
    &\log[((1-\alpha)m_j)!^{r-1}]
    = (r-1)\log\left[\prod_{i=1}^{(1-\alpha)m_j}i\right]
    = (r-1)\sum_{i=1}^{(1-\alpha)m_j}\log(i) \\
    \leq& (1-\alpha)(r-1)\sum_{i=1}^{m_j}\log(i)
    = (1-\alpha)\log(m_j!^{r-1})
    \end{align*}
    The inequality comes from that removing the $\alpha m_j$ largest terms from the sum reduces the sum by at least a fraction $\alpha$, since the logarithm is an increasing function.
\end{proof}

\subsection{Bounding the Knowledge of Individual Atoms}

We now create an upper bound on how many tuples in relation $S_j$ can be known by a fixed machine $c$. This lemma applies to any data distribution where the entropy function satisfies \autoref{entropy_property}.

\begin{lemma}
    \label{proportionalDiscovery}
    Assume we have a relation $S_j$ where the entropy function $\psi_{S_j}$ satisfies \autoref{entropy_property}. Then the following bound on the number of known tuples in $S_j$ holds.
    \[
        E[|K_{msg}^c(S_j)|]\leq \gamma_j f_{c,i}m_j
    \]
\end{lemma}

\begin{proof}
    Bound the entropy by the joint entropy and then use the entropy chain rule.
    \begin{align*}
        H(S_j)&\leq H(S_j,Msg_j)=H(Msg_j)+H(S_j|Msg_j) \\
        &= H(Msg_j(S_j))+\sum_{msg_j}Pr(Msg_j(S_j)=msg_j)H(S_j|Msg_j(S_j)=msg_j)
    \end{align*}
    Use that $H(Msg_j(S_j)$, the expected number of bits of information received about $S_j$, is less than the maximum length of $msg_j$, which is $f_{c,i}M_j$.
    \begin{align*}
        H(S_j)&\leq f_{c,i} M_j+\sum_{msg_j}Pr(Msg_j(S_j)=msg_j)\psi(m_j-|K_{msg}^c(S_j)|) \\
        &\leq f_{c,i} H(S_j)+\psi(m_j-|K_{msg}^c(S_j)|) \\
        &\leq f_{c,i} H(S_j)+\left( 1-\frac{|K_{msg}^c(S_j)|}{m_j\gamma_j} \right)\psi(m_j)\\
        &= f_{c,i} H(S_j)+\left( 1-\frac{|K_{msg}^c(S_j)|}{m_j\gamma_j} \right)H(S_j)
    \end{align*}
    The lemma follows from this.
\end{proof}

\subsection{Bounding the Knowledge of the Query}

In the remaining part of this section, we will assume that all relations satisfy \autoref{proportionalDiscovery}. For a potential tuple $a_j\in[n]^{r_j}$ in atom $S_j$, let $p_{c,i}(a_j)=Pr[a_j\in K_{msg}^c(S_j)]$. We give the following lemma.

\begin{lemma}
    \label{lemma:pjbounds}
    \begin{align*}
        p_{c,i}(a_j) &\leq \frac{m_j}{n^{r_j}} \\
        \sum_{[n]^{r_j}}p_{c,i}(a_j) &\leq \gamma_j f_{c,i}m_j
    \end{align*}
\end{lemma}
\begin{proof}
    The first statement is true because $p_{c,i}(a_j)\leq Pr(a_j\in S_j)=m_j/n^{r_j}$. The second statement follows straight from \autoref{proportionalDiscovery}, since $\sum_{[n]^{r_j}}p_{c,i}(a_j)=E[|K_{msg}(S_j)|]\leq\gamma f_j|S_j|$.
\end{proof}

We now give a lemma about the expected output size of the query.

\begin{lemma}
    Assume that for each relation $S_j$, the probability across tuples $a_j$ to appear in $S_j$ is uniform. Then the expected output size is the following.
    \[ 
    E[|q(I)|]=\left(\prod_{j\in[l]}m_j\right)n^{\textstyle k-\sum_{j\in[l]}r_j}
    \]
\end{lemma}
\begin{proof}
    \begin{align*}
    E[|q(I)|]=\sum_{a\in [n]^k}Pr(a\in q(I))
    =\sum_{a\in [n]^k}Pr\left(\bigwedge_{j\in[l]} a_j\in S_j\right)
    =\sum_{a\in [n]^k}\prod_{j\in[l]}Pr(a_j\in S_j)\\
    =\sum_{a\in [n]^k}\prod_{j\in[l]}\frac{m_j}{n^{r_j}}
    =\left(\prod_{j\in[l]}m_j\right)\sum_{a\in[n]^k}n^{-\textstyle\sum_{j\in[l]}r_j}
    =\left(\prod_{j\in[l]}m_j\right)n^{\textstyle k-\sum_{j\in[l]}r_j}
    \end{align*}
\end{proof}

\begin{definition}
    Given a query $q(\mathbf{x})=S_1(\mathbf{x_1}),...,S_l(\mathbf{x_l})$, the extended query $q'$ is defined by:
    \[
        q'(\mathbf{x}):-S_1(\mathbf{x_1}),...,S_l(\mathbf{S_l}),T_1(x_1),...,T_k(x_k)
    \]
    That is, for each variable $x_i$ in $q$ we add a unary atom $T_i(x_i)$.
\end{definition}

We can now prove the main lemma in this appendix, which bounds the expected number of known output tuples given that fractions $f_{c,i}$. 

\begin{proof}[Proof \autoref{lemma:LowerBoundLemma}]
    \[
    E[|K_{msg}^c(q(I))|]
    = \sum_{a\in[n]^k}Pr(a\in K_{msg}(q))
    = \sum_{a\in [n]^k}\prod_{j\in[l]}p_{c,i}(a_j)
    \]

    Let $q'$ be the extended query. We will create a fractional edge cover over $q'$ where the values of the edge cover on the edges that also exist in the original query are exactly $u_j$, the value in the edge packing. For each atom $T_i$, set the corresponding fractional edge cover value to $u_i'=1-\sum_{j:x_i\in S_j}u_j$. $\mathbf{(u,u')}$ is then both a tight fractional edge packing and a tight fractional edge cover.
    
    We will use Friedguts inequality to bound the above. Denote the set of atoms in $q\cap q'$ as $A$ and the newly added singleton atoms $S$. For all atoms in $S$, set the variable $p_{c,i}$ to 1. We will initially assume that all $u_j>0$.
    \begin{align*}
    E[|K_{msg}^c(q(I))|]
    =& \sum_{a\in [n]^k}\prod_{j\in A}p_{c,i}(a_j)
    =\sum_{a\in[n^k]}\left(\prod_{j\in A}p_{c,i}(a_j)\prod_{j\in S}p_{c,i}(a_j)\right) \\
    \leq& \prod_{j\in A}\left(\sum_{a_j\in[n]^{r_j}}p_{c,i}(a_j)^{1/u_j}\right)^{u_j} \prod_{j\in S}\left( \sum_{a_j\in[n]}p_{c,i}(a_j)^{1/u_j'} \right)^{u_j'} \\
    =& \prod_{j\in[l]}\left( \sum_{a_j\in[n]^{r_j}}p_{c,i}(a_j)^{1/u_j} \right)^{u_j} \prod_{i=1}^k n^{u_i'} \\
    \end{align*}
    Now use \autoref{lemma:pjbounds}.
    \begin{align*}
    E[|K_{msg}[q]|]\leq&
    \prod_{j\in[l]}\left( \sum_{a_j\in[n]^{r_j}}p_{c,i}(a_j)^{1/u_j-1}p_{c,i}(a_j) \right)^{u_j} \prod_{i=1}^k n^{u_i'}\\
    \leq&\prod_{j\in[l]}\left(\left(\frac{m_j}{n^{r_j}}\right)^{1/u_j-1} \gamma_jf_{c,i}m_j \right)^{u_j} \prod_{i=1}^k n^{u_i'} \\
    =& \prod_{j\in[l]}\left( \frac{\gamma_j f_{c,i} m_j^{1/u_j}}{n^{r_j(1/u_j-1)}} \right)^{u_j} \prod_{i=1}^k n^{u_i'} \\
    =& \left(\prod_{j\in[l]}m_j \prod_{j\in[l]}(\gamma_j f_{c,i})^{u_j}\right) n^{\textstyle \sum_{j\in[l]}r_j(u_j-1)+\sum_{i\in[k]}u_i'}
    \end{align*}
    We will rewrite the exponent of $n$ using the following, which uses the definition of $u_i'$.
    \[ \sum_{j\in[l]}r_ju_j+\sum_{i=1}^ku_i'=\sum_{i=1}^k\left(\sum_{j:x\in S_j}u_j+u_i'\right)=\sum_{i=1}^k1=k \]
    This finally gives.
    \[
    E[|K_{msg}[q]|]
    \leq \left( n^{\textstyle k-\sum_{j\in[l]}r_j} \prod_{j\in[l]}m_j \right)\prod_{j\in[l]}(\gamma_j f_{c,i})^{u_j}
    =E[|q(I)|]\prod_{j\in[l]}(\gamma_j f_{c,i})^{u_j}
    \]
\end{proof}

If for any $j$, $u_j=0$, we can set $u_j=\delta$ for some small $\delta$. Then let $\delta\rightarrow 0$.

\section{General Cost Function Upper Bound}
\label{appendix:GeneralCostUpperBound}

We will show the following main theorems.

\UpperBoundDenseGeneral*
\UpperBoundSparseGeneral*

We start by analyzing the partitioning of the space $\Lambda$.

\begin{lemma}
\label{lemma:GeneralCostUBPartition}
    Let $\bv=(v_1,...,v_k)$ be a minimal fractional vertex cover of $q(\mathbf{x})$. Let $L^*=\Bar{L}^{lower}$ be the maximal lower bound of $q$. Let the side length of machine $c$ along variable $x$ in the hypercube algorithm be picked according to earlier specified.
    Then the following two properties hold:
    \begin{enumerate}
        \item $\sum_{c\in[p]}|\Lambda_c|=n^k$
        \item for every atom $S_j$ with arity $r_j$: $|\pi_{S_j}\Lambda_c|\leq\left(\frac{g_c^*(L^*)}{m}\right) n^{r_j}$
    \end{enumerate}
\end{lemma}
\begin{proof}
    We start with the first claim.
    \begin{align*}
    \sum_{c\in[p]}\prod_{i\in[k]} \lambda_{c,i}
    = \sum_{c\in[p]}\prod_{i\in[k]} \left(\frac{g_c^*(L^*)}{m}\right)^{v_i}n
    = n^k\sum_{c\in[p]} \left(\frac{g_c^*(L^*)}{m}\right)^{v}
    = n^k \frac{\sum_{c\in[p]} g_c^*(L^*)^{v}}{m^v} \\
    \geq n^k \frac{m^v}{m^v}=n^k
    \end{align*}
    Next, we show the second claim.
    \begin{align*}
        |\pi_{S_j}\Lambda_c| &= \prod_{x\in S_j}\lambda_{c,i}
        = \prod_{i\in S_j}\left(\frac{g_c^*(L^*)}{m}\right)^{v_i}n
        \leq \left(\frac{g_c^*(L^*)}{m}\right) n^{r_j}
    \end{align*}
    The inequality comes from that $g_c^*(L^*)/m\leq 1$ combined with that $\bv$ is a vertex cover.
\end{proof}

The lemma above tells us how the subspaces $\{\Lambda_c\}_c$ should be dimensioned. We will similar to in the case with linear cost function need to geometrically position the subspaces to cover $\Lambda$. We will use the same method as previously described in \autoref{subsec:Packing}. We will need to show the following lemma, which is similar to \autoref{lemma:machineordering}.

\begin{lemma}
    Let $c,c'\in[p]$ be two machines. Then, if there exists a variable $z\in[k]$ such that $\lambda_{c,z}\leq \lambda_{c',z}$,
    then for all variables $i\in[k]$, $\lambda_{c,i}\leq \lambda_{c',i}$.
\end{lemma}
\begin{proof}
    Remember that the side of $\Lambda_c$ follows the following expression.
    \[
        \lambda_{c,i}=\left(\frac{g_c^*(L^*)}{m}\right)^{v_i}n
    \]
    Since the vertex cover, $m$, $n$ and $L^*$ are the same between expressions for $\lambda_{c,z}$ and $\lambda_{c',z}$, we know that if for some $z$, $\lambda_{c,z}\leq \lambda_{c',z}$, then $g_c^*(L^*)\leq g_c^*(L^*)$. This means that for all $i$, $\lambda_{c,i}\leq \lambda_{c',i}$. This holds also after the sides of subspaces have been rounded to the nearest power of two.
\end{proof}

This lemma implies that the method for packing described in \autoref{subsec:Packing} also works in this case. We can now show the main theorems.

\begin{proof}[Proof \autoref{theorem:UpperBoundDenseGeneral}]
    The number of tuples that machine $c$ receives is the number of tuples in $\Lambda_c$. In the worst case, every tuple that could exist in $\Lambda_c$ does exist. We will calculate the load $L_{cj}$ on machine $c$ from relation $S_j$.
    $$
    L_{cj}=g_c(|\pi_{S_j}\Lambda_c|\log n)\leq g_c\left(\frac{g_c^*(L^*)}{m}n^r\log n\right)
    $$
    Since $n^r\log n=m/\theta \cdot \log n$, this gives
    \[
    L_{cj}\leq g_c\left(\frac{g_c^*(L^*)\log n}{\theta}\right)
    \]
    Note $\log n/\theta \geq 1$ since $\theta\leq1$ and $\log n\geq 1$. Define $\delta$ such that $(1+\delta)=\log n/\theta$.
    \[
    L_{cj}\leq g_c\left((1+\delta)g_c^*(L^*)\right)
    =(1+\delta)^{a_i}g_c(g_c^*(L^*))
    =(\log n/\theta)^{a_i}g_c(\max_xg_c(x)\leq L^*)
    \leq \tilde{O}(L^*)
    \]
\end{proof}

\begin{proof}[Proof \autoref{theorem:UpperBoundSparseGeneral}]
    Denote $N_{c,j}$ as the number of bits mapped to machine $c$ and $n_{c,j}$ as the number of tuples mapped to the same machine, for some fixed relation $S_j$. Let $L_{cj}$ be the corresponding load. The probability that a tuple $t\in S_j$ maps to machine $c$ is the following, by \autoref{lemma:GeneralCostUBPartition}.
    \[
    Pr[(\pi_{S_j}\mathbf{h})(a_j)\in \Lambda_c]
    =\frac{|\pi_{S_j}\Lambda_c|}{n^{r_j}}
    \leq\frac{g_c^*(L^*)}{m}
    \]
    Since the probability is uniform for each tuple, $n_{c,j}\sim Bin\left(m, \frac{g_c^*(L^*)}{m}\right)$. Note the following
    \begin{align*}
    g_c(E[n_{c,j}]\log n)
    =g_c\left(\frac{g_c^*(L^*)}{m}m\log n\right)
    =g_c(g_c^*(L^*)\log n)
    \end{align*}
    Here we used that the expected value of a random variable with distribution $Bin(n,p)$ is $n\cdot p$. Define $\delta'$ so $1+\delta'=\log n$. Then 
    \[ 
    g_c(g_c^*(L^*)\log n)
    =g_c(g_c^*(L^*)(1+\delta'))
    \leq(1+\delta')^{a_i}g_c(g_c^*(L^*))
    \leq (\log n)^{a_i}L^*
    \] 
    We can use this to show that the expected load matches the lower bound.
    \[
    E[L_{cj}]= g_c(E[n_{c,j}]\log n) = \tilde{O}(L^*)
    \]
    Next, we bound the probability $P$ that the load is more than a factor $(1+\epsilon)$ more than $(\log n)^{a_i}L^*$, where $a_i$ is the constant in \autoref{def:CostFunction} for cost function $g_c$. This will prove the lemma.
    \[
    P=Pr[L_{cj}\geq (1+\epsilon)(\log n)^{a_i}L^*]
    \leq Pr[L_{cj}\geq (1+\epsilon)g_c(g_c^*((\log n)^{a_i}L^*))]
    \]
    Define $\delta$ such that $(1+\delta)=(1+\epsilon)^{1/a_i}$. This gives.
    \begin{align*}
    P &\leq Pr[L_{cj}\geq (1+\delta)^{a_i} g_c(g_c^*((\log n)^{a_i}L^*))] \\
    &\leq Pr[L_{cj}\geq g_c((1+\delta)g_c^*((\log n)^{a_i}L^*))]
    =Pr[N_{c,j}\geq (1+\delta)g_c^*((\log n)^{a_i}L^*)]
    \end{align*}
    Here the second inequality comes from the third property in \autoref{def:CostFunction}. The last equality comes from that $L_{cj}=g_c(N_{c,j})$. Recall that $(\log n)^{a_i}L^*\geq g_c(E[n_{cj}]\log n)=g_c(E[N_{c,j}])$, as showed before. This gives $g_c^*((\log n)^{a_i}L^*) \geq E[N_{c,j}]$.
    \[
    P \leq Pr[N_{c,j}\geq (1+\delta)E[N_{c,j}]]
    = Pr[n_{c,j} \geq (1+\delta)E[n_{c,j}]]
    \]
    Finally, use the Chernoff bound.
    \[
    P \leq Pr[n_{c,j}\geq (1+\delta)E[n_{c,j}]]
    \leq \exp\left( -\delta^2\frac{E[n_{c,j}]}{3} \right)
    =\exp\left(-((1+\epsilon)^{1/a_i}-1)^{2}\frac{E[n_{c,j}]}{3}\right)
    \]
    The probability that the load exceeds $\tilde{O}(L^*)$ is obtained using the union bound across all atoms and machines.
\end{proof}

\section{Unequal Cardinality Upper Bounds}
\label{appendix:UnEqualCardinality}

In this appendix, we complete the upper bound for the cartesian product, binary join, star query and triangle query in \autoref{sec:unequal}. The techniques presented can be used on other queries as well. We begin by describing how to dimension the sides of the subspace $\Lambda_c$ for each machine. Then we give proofs of the following two main theorems.

\UpperBoundDiffCardinalityDense*
\UpperBoundDiffCardinalitySparse*

\subsection{Partitioning the Space}

\subsubsection{Cartesian Product}

We consider the cartesian product:
\[ q(x,y)\dl S_1(x),S_2(y) \]

\begin{lemma}
    \label{lemma:cartesianEPBound}
    Consider the cartesian product. Let $L^*$ be the lower bound on the load and let $\{\bu_c\}_c$ be the set of edge packings that give the lower bound $L^*$. Then for both atoms $j\in\{1,2\}$, the following holds:
    \[ 
    \left(\frac{L^*w_c}{M_j}\right)^{u_{c,j}}
    =\min\left(\frac{L^*w_c}{M_j}, 1\right)
    \leq\frac{L^*w_c}{M_j} 
    \]
\end{lemma}
\begin{proof}
    The edge packing $\mathbf{u_i}=(u_{i,1},u_{i,2})$ is set to minimize $\prod_{j\in[l]}\left(\frac{L^*w_c}{M_j}\right)^{u_{c,j}}$ where $\mathbf{u_i}$ is an edge packing. For the cartesian product, the two values $(u_{i,1},u_{i,2})$ can be set independently. Therefore, $u_{c,j}=1$ if $L^*w_c/M_j\leq1$, and otherwise $u_{c,j}=0$. The lemma follows.
\end{proof}

\begin{theorem}
    Consider the cartesian product. Let $L^*$ be the lower bound on the load, and $\bu_c=(u_{c,x},u_{c,y})$ be the corresponding edge packing for machine $c$. Let the side length of $\Lambda_c$ along dimension $j$ be
    \[ \lambda_{c,i}:=\min\left(\frac{L^*w_c}{M_j}, 1\right)n \]
    Then we have the following two properties:
    \begin{enumerate}
        \item $\sum_{c\in[p]}|\Lambda_c|\geq V[\Lambda]$
        \item for every atom $S_j$ with arity $r_j$ and cardinality $M_j$, and machine $c$:
        $\pi_{S_j}\Lambda_c\leq\frac{L^*w_c}{M_j}n^{r_j}$
    \end{enumerate}
\end{theorem}
\begin{proof}
    We use \autoref{lemma:cartesianEPBound} for proving both properties. Start with the first first property. The inequality comes from that $L^*$ satisfies the lower bound.
    \[
    \sum_{c\in[p]} \lambda_{c,i}\lambda_{c,y}
    =\left(\sum_{c\in[p]}\prod_{j\in[l]}\min\left(\frac{L^*w_c}{M_j}, 1\right)\right)n^2
    =\left(\sum_{c\in[p]}\prod_{j\in[l]}\left( \frac{L^*w_c}{M_j} \right)^{u_{c,j}}\right)n^2
    \geq n^2=V[\Lambda]
    \]
    Now we show the second property. 
    \[
    |\pi_{S_j}\Lambda_c|=\lambda_{c,i}=\min\left(\frac{L^*w_c}{M_j}, 1\right)n
    \leq\frac{L^*w_c}{M_j}n
    \]
\end{proof}

\begin{lemma}
    For any two machines in the cartesian product with weights $w_c$ and $w_{c'}$ such that $w_c\leq w_{c'}$, for any variable $x_i \in[k]$, $\lambda_{c,i}\leq\lambda_{c',i}$.
\end{lemma}
\begin{proof}
    $\lambda_{c,i}=\min\left(L^*w_c/M_j,1\right)n$ is an increasing function with $w_c$. The lemma follows from this.
\end{proof}

\subsubsection{Binary Join}
Next, consider the binary join.
\[ q(x,y,z)\dl S_1(x,z),S_2(y,z) \]
Without loss of generality, assume $|S_1|\geq|S_2|$.

\begin{lemma}
\label{lemma:BinaryJoin10}
    Consider the binary join. The fractional edge packing $(u_1,u_2)=(1,0)$ always gives the optimal lower bound $L^*$.
\end{lemma}
\begin{proof}
    Denote $s_{\bu}=\prod_{j\in[l]}(L^*w_c/M_j)^{u_i}$ as the term in \autoref{equation:DiffCardLowerBound} for machine $c$ and edge packing $\bu$. The edge packing is picked by minimizing $s_{\bu}$. Consider any fractional edge packing $(a,b)$, $a<1$. We will show that $s_{\Bar{u}}\leq s_{(a,b)}$, which proves the lemma. We know $a+b\leq 1$ because of the edge packing constraint on variable $z$. This gives.
    \[
    s_{(a,b)} = \left(\frac{L^*w_c}{M_1}\right)^a\left(\frac{L^*w_c}{M_2}\right)^b
    \geq \left(\frac{L^*w_c}{M_1}\right)^{a+b}
    \geq \left(\frac{L^*w_c}{M_1}\right)^1 = s_{(1,0)}
    \]
\end{proof}

\begin{theorem}
    Consider the join of two relations. Let $L^*$ be the lower bound on the load. Let the side lengths of $\Lambda_c$ be the following:
    \[ 
    \lambda_{c,x}:=n, \quad \lambda_{c,y}:=n, \quad
    \lambda_{c,z}:=\frac{L^*w_c}{M_1}n 
    \]
    Then we have the following two properties:
    \begin{enumerate}
        \item $\sum_{c\in[p]}|\Lambda_c|\geq V[\Lambda]$
        \item for every atom $S_j$ with arity $r_j$ and cardinality $M_j$, and machine $c$:
        $\pi_{S_j}\Lambda_c\leq\frac{L^*w_c}{M_j}n^{r_j}$
    \end{enumerate}
\end{theorem}
\begin{proof}
    We begin with the first claim. The inequality comes from that $L^*$ satisfies \autoref{equation:DiffCardLowerBound}, since by \autoref{lemma:BinaryJoin10} is always an optimal edge packing.
    \[
    \sum_{c\in[p]}\lambda_{c,i}\lambda_{c,y}\lambda_{c,z}
    =\left(\sum_{c\in[p]}\frac{L^*w_c}{M_1}\right)n^3
    \geq n^3=V[\Lambda]
    \]
    Next, we show the second claim.
    \[
    \pi_{S_j}\Lambda_c=\lambda_{c,z}n=\frac{L^*w_c}{M_1}n^2\leq\frac{L^*w_c}{M_j}n^2
    \]
\end{proof}

\begin{lemma}
    For any two machines in the join with weights $w_1$ and $w_2$ such that $w_1\leq w_2$, for any variable $j\in[k]$, $\lambda_{1,j}\leq\lambda_{2,j}$.
\end{lemma}
\begin{proof}
    The lemma follows from that for all $j\in\{x,y,z\}$, $\lambda_{c,i}$ is an increasing function with $w_c$.
\end{proof}

\subsubsection{Star Query}
Next, we consider the star query.
\[ q(z,x_1,\dots,x_a)\dl S_1(z,x_1),\dots,S_n(z,x_a) \]
Without loss of generality, assume that $S_1$ is the atom with maximal cardinality.

\begin{lemma}
    Consider the star query. The fractional edge packing $\Bar{u}$ where $\Bar{u}_1=1$ and for all other $j$, $\Bar{u}_j=0$, always give the optimal lower bound $L^*$.
\end{lemma}
\begin{proof}
    Denote $s_{\bu}=\prod_{j\in[l]}(L^*w_c/M_j)^{u_j}$ as the term in \autoref{equation:DiffCardLowerBound} for machine $c$ and any edge packing $\bu$. The edge packing is picked by minimizing $s_{\bu}$. We will show that $s_{\Bar{u}}\leq s_{\bu}$. We know $\sum_{j\in[l]}u_j\leq 1$, because of the edge packing constraint on variable $z$. This gives
    \[
    s_{\bu}=\prod_{j\in[l]}\left(\frac{L^*w_c}{M_j}\right)^{u_j}
    \geq \prod_{j\in[l]}\left(\frac{L^*w_c}{M_1}\right)^{u_j}
    =\left(\frac{L^*w_c}{M_1}\right)^{\sum_{j\in[l]}u_j}
    \geq \frac{L^*w_c}{M_1}
    =s_{\Bar{u}}
    \]
    Here the first inequality uses that $M_1$ is the maximum cardinality.
\end{proof}

\begin{theorem}
    Consider the star query. Let $L^*$ be the lower bound on the load. Let the side lengths of $\Lambda_c$ be the following:
    \[
    \lambda_{c,z}=\frac{L^*w_c}{M_1}n,
    \quad \forall i\in[a]\; \lambda_{c,x_i}=n
    \]
    Then we have the following two properties:
    \begin{enumerate}
        \item $\sum_{c\in[p]}|\Lambda_c|\geq V[\Lambda]$
        \item for every atom $S_j$ with arity $r_j$ and cardinality $M_j$, and machine $c$:
        $\pi_{S_j}\Lambda_c\leq\frac{L^*w_c}{M_j}n^{r_j}$
    \end{enumerate}
\end{theorem}
\begin{proof}
    Start by showing the first claim. Similar to before we use \autoref{equation:DiffCardLowerBound} to obtain the inequality.
    \[ \sum_{c\in[p]}|\Lambda_c|=\sum_{c\in[p]}\lambda_{c,z}n^{a}=\sum_{c\in[p]}\frac{L^*w_c}{M_1}n^{a+1}\geq n^{a+1}=V[\Lambda]
    \]
    Next, we show the second claim. We will use that $M_1$ is the maximum cardinality among relations.
    \[
    \pi_{S_j}\Lambda_c=\lambda_{c,z}n=\frac{L^*w_c}{M_1}n^2
    \leq\frac{L^*w_c}{M_j}n^2
    \]
\end{proof}

\begin{lemma}
    For any two machines in the join with weights $w_c$ and $w_{c'}$ such that $w_c\leq w_{c'}$, for any variable $x_i \in[k]$, $\lambda_{c,i}\leq\lambda_{c',i}$.
\end{lemma}
\begin{proof}
    This lemma follows from that for all $i\in[k]$, $\lambda_{c,i}$ is an increasing function with $w_c$.
\end{proof}

\subsubsection{Triangle Query}
Next, consider the triangle query. Without loss of generality, suppose $M_1\geq M_2\geq M_3$.
\[
q(x,y,z)\dl S_1(x,y),S_2(y,z),S_3(z,x)
\]
Let $L^*$ be the lower bound on the load. We will introduce the following quantities.
\begin{align*}
    f_{c,x}=\sqrt{L^*w_c\frac{M_2}{M_1M_3}} \\
    f_{c,y}=\sqrt{L^*w_c\frac{M_3}{M_1M_2}} \\
    f_{c,z}=\sqrt{L^*w_c\frac{M_1}{M_2M_3}}
\end{align*}

Note that for a machine $c$, $f_{c,z}\geq f_{c,x}\geq f_{c,y}$, because of the order we assumed on the cardinalities. For sufficiently small machines, all quantities above are less than 1. As $w_c$ increases, some quantities might become equal to or more than $1$. We will classify machines into three different sets based on which of the quantities are greater than $1$ and which are less than $1$.

\begin{definition}
    Label machine $c$ as \textit{small}, \textit{medium} or \textit{big} according to the following.
    \begin{itemize}
        \item \textit{small} if $1 > f_{c,z} \geq f_{c,x} \geq f_{c,y}$
        \item \textit{medium} if $f_{c,z} \geq 1 > f_{c,x} \geq f_{c,y}$
        \item \textit{big} if $f_{c,z} \geq f_{c,x} \geq 1 > f_{c,y}$
    \end{itemize}
\end{definition}

\begin{lemma}
    \label{lemma:smallEPOptimal}
    For small machines, the fractional edge packing that gives the lower bound is $(u_1,u_2,u_3)=(1/2,1/2,1/2)$.
\end{lemma}
\begin{proof}
    We will use that $f_{c,x},f_{c,y},f_{c,z} < 1$. For $f_{c,y}$, this gives
    \[
    \sqrt{L^*w_c\frac{M_3}{M_1M_2}} < 1
    \]
    This can be rewritten as
    \[
    \frac{(L^*w_c)^{3/2}}{\sqrt{M_1M_2M_3}} < \frac{L^*w_c}{M_3}
    \]
    Since the edge packing $\bu$ is picked by minimizing $\prod_{j\in[l]}(L^*w_c/M_j)^{u_{c,j}}$, the above means that the edge packing $(1,1,1)$ gives a lower bound than the edge packing $(0,0,1)$. By a similar argument for $f_{c,x}$ and $f_{c,z}$, the edge packings $(0,1,0)$ and $(0,0,1)$ also give lower bounds.
\end{proof}

\begin{lemma}
    \label{lemma:medBigEPOptimal}
    For medium and big machines, the fractional edge packing that gives the lower bound is $(u_1,u_2,u_3)=(1,0,0)$.
\end{lemma}
\begin{proof}
    We will use an argument similar to in the previous lemma. In this case, we know that $f_{c,y}\geq 1$
    \[
    \sqrt{L^*w_c\frac{M_3}{M_1M_2}} \geq 1
    \]
    This can be rewritten as
    \[
    \frac{(L^*w_c)^{3/2}}{\sqrt{M_1M_2M_3}}\geq\frac{L^*w_c}{M_3}
    \]
    By a similar argument as in the previous lemma, this means that the edge packing $(1,0,0)$ gives a higher lower bound than $(1,1,1)$. Since still $f_{c,x},f_{c,z}\leq1$, $(0,0,1)$ and $(0,1,0)$ gives even lower bounds than $(1,1,1)$.
\end{proof}

We will now show how to pick the dimension of the subspaces allocated to each machine. We will define the function $g_{\Lambda}:\mathbb{R}_{>0}\rightarrow(\mathbb{R}_{>0},\mathbb{R}_{>0},\mathbb{R}_{>0})$ that maps from the weight $w_c$ of a machine to the vector of dimensions $(\lambda_{c,z},\lambda_{c,y},\lambda_{c,z})$ of the subspace for the machine.

\[
    g_{\Lambda}(w_c)=
    \begin{cases}
        (f_{c,x}n,f_{c,y}n,f_{c,z}n) & \text{if $w_c$ is small} \\
        (f_{c,x}n,f_{c,y}n,n) & \text{if $w_c$ is medium} \\
        (n,\frac{L^*w_c}{M_1}n,n) & \text{if $w_c$ is big} \\
    \end{cases}
\]

\begin{lemma}
    Set sides of $\Lambda_c$ according to the above. Then:
    \begin{equation*}
        \sum_{c\in[p]}|\Lambda_c|\geq V[\Lambda]
    \end{equation*}
\end{lemma}
\begin{proof}
    We will consider two types of terms in the sum, the terms corresponding to small machines and medium/large machines. For small machines, the term is
    \[
    |\Lambda_c|=\lambda_{c,x}\lambda_{c,y}\lambda_{c,z}
    =\frac{(L^*w_c)^{3/2}}{\sqrt{M_1M_2M_3}}n^3
    \]
    For medium and large machines, the term is the following
    \[
    |\Lambda_c|=\lambda_{c,x}\lambda_{c,y}\lambda_{c,z}=\frac{L^*w_c}{M_3}n^3
    \]
    Observe that both these terms can be written as $\prod_{j\in[l]}(L^*w_c/M_j)^{u_{c,j}}n^3$, where $\bu_c$ is the edge packing giving the highest lower bound for machine $c$. This is because, by the previous lemma, $\bu_c=(1,1,1)$ for small machines and $\bu=(1,0,0)$ for medium and large machines. Therefore we get
    \[
    \sum_{c\in[p]}|\Lambda_c|=\sum_{c\in[p]}\prod_{j\in[l]}\left(\frac{L^*w_c}{M_j}\right)^{u_{c,j}}n^3\geq n^3=V[\Lambda]
    \]
    Here the inequality uses \autoref{equation:DiffCardLowerBound}.
\end{proof}

\begin{theorem}
    Consider the triangle query. If sides of $\Lambda_c$ are set as described earlier, the following is true:
    \[ \pi_{S_j}\Lambda_c\leq\frac{L^*w_c}{M_j}n^{r_j} \]
\end{theorem}
\begin{proof}
    We start with showing this for small and medium machines together, and then finally for large machines.
    
    \textbf{Small and medium machines:} 
    For small and medium machines $\lambda_{c,x}=n\cdot\min\{f_{c,x},1\}$. Note that $\min\{f_{c,x},1\}\leq f_{c,x}$. This gives:
    \[
    \pi_{S_j}\Lambda_c
    =\prod_{j\in S_j}\lambda_{c,x}
    =\prod_{j\in S_j}\min\{f_{c,i},1\}n
    \leq\prod_{j\in S_j}f_{c,i}n
    =\frac{L^*w_c}{M_j}n^r
    \]

    \textbf{Big machines:} 
    We show this for each relation at a time. Start with $S_3$. Use that both $f_{c,x}$ and $f_{c,z}$ are more than 1.
    \[
    \pi_{S_3}\Lambda_c=\lambda_{c,z}\lambda_{c,x}= n^2\leq f_{c,x}f_{c,z}n^2=\frac{L^*w_c}{M_3}n^2
    \]
    Now $S_2$.
    \[
    \pi_{S_2}\Lambda_c=\lambda_{c,y}\lambda_{c,z}=\frac{L^*w_c}{M_1}n^2\leq\frac{L^*w_c}{M_2}n^2
    \]
    Finally for $S_1$.
    \[
    \pi_{S_1}\Lambda_c=\lambda_{c,y}\lambda_{c,x}=\frac{L^*w_c}{M_1}n^2
    \]
\end{proof}

\begin{lemma}
    For any two machines in the triangle query with weights $w_c$ and $w_{c'}$ such that $w_c\leq w_{c'}$, for any variable $x_i \in[k]$, $\lambda_{c,i}\leq\lambda_{c',i}$.
\end{lemma}
\begin{proof}
    We have to prove that $g_{\Lambda}$ is an increasing function in all three output dimensions. We can see that if the two machines have the same label, either small, medium, or big, this property holds. This is because for all three labels, all three dimensions are increasing functions of $w_c$. It remains to argue that this property holds for machines with different labels.

    We will do this by arguing that at the two points where a machine switches from \textit{small} to \textit{medium}, and \textit{medium} to \textit{large}, $g_{\Lambda}$ is continuous. When a machine switches from \textit{small} to \textit{medium}, $(f_{c,x}n,f_{c,y}n,f_{c,z}n)=(f_{c,x}n,f_{c,y}n,n)$, since $f_{c,z}=1$. Similarly, a machine switches from \textit{medium} to \textit{large} when $f_{c,x}=1$. Then $(f_{c,x}n,f_{c,y}n,n)=(n,\frac{L^*w_c}{M_1}n,n)$, since $f_{c,x}=\sqrt{L^*w_cM_2/(M_1M_3)}=1$, can be rewritten as $L^*w_c/M_1=f_{c,y}$. This proves the lemma.
\end{proof}

\subsection{Matching the Lower Bound}

We can now prove the main theorems, which show that our algorithm matches the lower bound for the aforementioned queries for both data distributions.

\begin{proof}[Proof \autoref{theorem:UpperBoundDiffCardinalityDense}]
    The worst-case load is obtained if every possible tuple that would be sent to the subspace $\Lambda_c$ is present in the input database instance. That gives the following load:
    \[
    L_{cj}=\frac{\log n}{w_c}|\pi_{S_j}\Lambda_c|
    \leq\frac{n^r\log n}{w_c}\frac{L^*w_c}{M_j}
    =O(L^*\log n)
    \]
\end{proof}

\begin{proof}[Proof \autoref{theorem:UpperBoundDiffCardinalitySparse}]
    The probability that a tuple in $\mathbf{a_j} \in S_j$ maps to machine $c$ is the following:
    $$ Pr[(\pi_{S_j}\mathbf{h})a_j \in \Lambda_c]
    = \frac{\pi_{S_j}\Lambda_c}{n^{r_j}}
    \leq \frac{L^*w_c}{M_j}
    $$
    We can now show that the expected load matches the lower bound. Use that $M_j=\Omega(m_j)$ by \autoref{lemma:EntropyOMJSparse}.
    $$
    E[L_{cj}] =\frac{m_j\log n}{w_c}\cdot Pr[(\pi_{S_j}\mathbf{h})a_j \in \Lambda_c]
    \leq \frac{m_j\log n L^*w_c}{w_cm_j}=O(L^*\log n)
    $$
    Furthermore, the number of tuples from $S_j$ mapped to machine $c$, $n_{c,j}$ has a binomial distribution, $n_{c,j}\sim Bin(m_j,L^*w_c/M_j)$. We can bound the probability with a Chernoff bound.
    \begin{align*}
    Pr[L_{cj}\geq(1+\delta)L^*]
    =Pr[N_{c,j}\geq(1+\delta)L^*w_c]
    =Pr[n_{c,j}\geq(1+\delta)\frac{L^*w_c}{M_j}m_j] \\
    \leq\exp\left( -\delta^2\frac{L^*w_c}{3M_j}m_j \right)
    \end{align*}
\end{proof}

The probability that the load exceeds $L^*$ for any machine $c$ and atom $S_j$ can be obtained using the union bound with the above probability.

\end{document}